\documentclass{article}
\usepackage{times}
\usepackage{soul}
\usepackage{url}

\usepackage[small]{caption}
\usepackage{graphicx}
\usepackage{amsmath}
\usepackage{booktabs}
\usepackage{amssymb,amsthm,dsfont}
\usepackage{color,soul}
\usepackage{xspace}
\usepackage{enumerate}
\usepackage{bbm}
\usepackage{mathtools}
\usepackage{comment}
\usepackage{paralist}
\usepackage{nicefrac}
\usepackage{multirow}
\usepackage{paracol}
\usepackage{fancyvrb}

\usepackage{subcaption}
\usepackage[noend,ruled,linesnumbered]{algorithm2e}
\usepackage{xifthen}

\usepackage[linecolor=green!70!black, backgroundcolor=green!10, bordercolor=black, textsize=small,disable]{todonotes}
\usepackage{fullpage}

\definecolor{winered}{rgb}{0.6,0.1,0.1}

\usepackage[%
hidelinks,
pagebackref,
colorlinks=true,
linkcolor=blue,
citecolor=winered,
urlcolor=purple,
pdfpagelayout=SinglePage,
pdfstartview=Fit]{hyperref}
\renewcommand*{\backref}[1]{}
\renewcommand*{\backrefalt}[4]{%
\ifcase #1%
\marginpar{\tiny no cite}
\or
 $\rightarrow$~p.~#2.%
\else
  $\rightarrow$~pp.~#2.%
\fi
}
\usepackage[numbers]{natbib}

\definecolor{darkgreen}{rgb}{0.01,0.6,0.1}
\newcommand{\myemph}[1]{{\color{darkgreen!70!black}\emph{#1}}}

\usepackage{tikz}
\usetikzlibrary{decorations,arrows,petri,topaths,backgrounds,shapes,positioning,fit,calc,decorations.pathreplacing,patterns,intersections,decorations.pathmorphing,matrix}

\tikzstyle{agentnode} = [inner sep=.5pt, circle, draw]

\makeatletter
\newcommand{\gettikzxy}[3]{%
  \tikz@scan@one@point\pgfutil@firstofone#1\relax
  \edef#2{\the\pgf@x}%
  \edef#3{\the\pgf@y}%
}
\makeatother

\usepackage{etoolbox} 

\newcommand{\appsymb}{$\star$}
\newcommand{\appref}[1]{{\hyperref[#1]{\appsymb}}}

\usepackage{cleveref}
\newtheorem{lemma}{Lemma}
\newtheorem{claim}{Claim}
\newtheorem{proposition}{Proposition}
\newtheorem{theorem}{Theorem}

\theoremstyle{definition}
\newtheorem{definition}[lemma]{Definition}
\newtheorem{example}{Example}

\crefname{table}{Table}{Tables}
\crefname{figure}{Figure}{Figures}
\crefname{theorem}{Theorem}{Theorems}
\crefname{definition}{Definition}{Definitions}
\crefname{corollary}{Corollary}{Corollaries}
\crefname{observation}{Observation}{Observations}
\crefname{lemma}{Lemma}{Lemmas}
\crefname{example}{Example}{Examples}
\crefname{reduction}{Reduction}{Reductions}
\crefname{construction}{Construction}{Constructions}
\crefname{subsection}{Subsection}{Subsections}
\crefname{section}{Section}{Sections}
\crefname{proposition}{Proposition}{Propositions}
\crefname{algorithm}{Algorithm}{Algorithms}
\crefname{drule}{Rule}{Rules}
\crefname{claim}{Claim}{Claims}
\crefname{appendix}{Appendix}{Appendix}

\newcommand{\myparagraph}[1]{%
\smallskip
\noindent \textbf{#1}%
}

\newcommand{\occ}{\ensuremath{\mathsf{occ}}}

\newcommand{\threesat}{\textsc{3-SAT}}

\newcommand{\ttrue}{\texttt{true}}
\newcommand{\tfalse}{\texttt{false}}

\newcommand{\ff}{\ensuremath{f}}
\newcommand{\negv}{\ensuremath{\overline{v}}}

\tikzstyle{blueline} = [thick, blue, dotted]
\tikzstyle{redline} = [thick, red, dashed]
\tikzstyle{blackline} = [thick, black]

\newcommand{\RO}{\textsc{Reachable Object}}
\newcommand{\ROs}{\RO}

\begin{document}
\sloppy

\newcommand{\mytitle}{Good Things Come to Those Who Swap Objects on Paths}
\title{\mytitle}

\newcommand{\papernumber}{Paper \#3817}
 \author{Matthias Bentert$^{1}$ \and Jiehua Chen$^{2}$ \and Vincent Froese$^1$
  \and Gerhard J.\ Woeginger$^3$\\
 {\small $^1$Algorithmics and Computational Complexity, Faculty~IV, TU Berlin, Berlin, Germany}\\
 {\small \texttt{\{matthias.bentert,vincent.froese\}@tu-berlin}}\\
 {\small $^2$University of Warsaw, Warsaw, Poland}\\
 {\small \texttt{jiehua.chen2@gmail.com}}\\
 {\small $^3$RWTH Aachen University, Aachen, Germany}\\
 {\small \texttt{woeginger@algo.rwth-aachen.de}}\\
 }

\date{}

\maketitle

\begin{abstract}
We study a simple exchange market, introduced by Gourv\`{e}s, Lesca and Wilczynski~(IJCAI-17), where every agent initially holds a single object.
The agents have preferences over the objects, and two agents may swap their 
objects if they both prefer the object of the other agent.
The agents live in an underlying social network that governs the structure of the swaps:
Two agents can only swap their objects if they are adjacent.
We investigate the {\RO} problem, which asks whether a given starting situation can 
ever lead, by means of a sequence of swaps, to a situation where a given agent obtains a given object.
Our results answer several central open questions on the complexity of {\RO}.
First, the problem is polynomial-time solvable if the social network is a path.
Second, the problem is NP-hard on cliques and generalized caterpillars.
Finally, we establish a three-versus-four dichotomy result for preference lists of 
bounded length:
The problem is easy if all preference lists have length at most three,
and the problem becomes NP-hard even if all agents have preference lists of length at most four.
  
  
\end{abstract}

\looseness=-1
\section{Introduction}\label{sec:intro}

Resource allocation under preferences is a widely-studied problem arising in areas such as artificial intelligence and economics.
We consider the case when resources are \emph{indivisible objects}
and each agent, having preferences over the objects, is to receive exactly one object.
In the standard scenario known as \emph{housing market}, each agent initially holds an object, and the task is to \emph{reallocate} the objects so as to achieve some desirable properties, such as Pareto optimality, fairness, or social welfare~\cite{shapley_cores_1974,roth_incentive_1982,AbrCecManMeh2005,SoeUen2010}.
While a large body of research in the literature takes a \emph{centralized} approach that globally controls and reallocates an object to each agent, 
we pursue a \emph{decentralized} (or \emph{distributed}) strategy where any pair of agents may locally \emph{swap} objects as long as this leads to an improvement for both of them, i.e., they both receive a more preferred object~\cite{damamme_power_2015}. 

To capture the situation where not all agents are able to communicate and swap with each other,
\citet{GouLesWil2017} introduced a variant of distributed object reallocation where the agents are embedded in an underlying social network so that agents can swap objects with each other only if \begin{inparaenum}[(i)] \item they are directly connected (socially tied) via the network and \item will be better off after the swap. \end{inparaenum}
To study the distributed process of swap dynamics along the underlying network topology, the authors analyzed various computational questions.
In particular, they study the \RO{} problem, which asks whether a given agent can reach a desired object via some sequence of mutually profitable swaps between agents.

Consider the following example (initial objects are drawn in boxes). 
If the underlying graph is complete, object~$x_3$ is reachable for agent~$1$ within one swap.
However, if the graph is a cycle as shown below, then to let object~$x_3$ reach
agent~$1$, agent~$3$ can swap with agent~$2$, and then agent~$2$ can swap object~$x_3$ with agent~$1$.
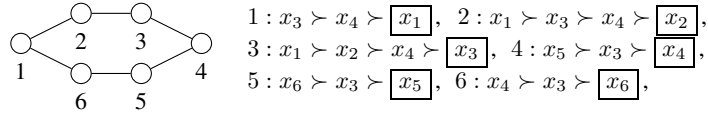
\begin{center}
 \begin{tikzpicture}[scale=1, every node/.style={scale=0.9}]
\node[circle,draw, label=below:1, inner sep=3pt] at (0,0) (1) {};
\node[circle,draw, label=below:2, inner sep=3pt] at (.8,.4) (2) {};
\node[circle,draw, label=below:3, inner sep=3pt] at (1.6,.4) (3) {};
\node[circle,draw, label=below:4, inner sep=3pt] at (2.4,0) (4) {};
\node[circle,draw, label=below:5, inner sep=3pt] at (1.6,-.4) (5) {};
\node[circle,draw, label=below:6, inner sep=3pt] at (0.8,-.4) (6) {};

\foreach \i / \j in {1/2, 2/3,3/4,4/5,5/6,6/1} {
  \draw (\i) -- (\j);
}
\node  at (4.3,.3) (v1) {$1:x_3\succ x_4 \succ$ \fbox{$x_1$}\,,};
\node[right = 0pt of v1] (v2) {$2:x_1 \succ x_3 \succ x_4 \succ$ \fbox{$x_2$}\,,};
\node[below = 12pt of v1.west,anchor=west] (v3) {$3:x_1 \succ x_2\succ x_4 \succ$ \fbox{$x_3$}\,,};
\node[right = -1pt of v3] (v4) {$4: x_5\succ x_3 \succ$ \fbox{$x_4$}\,,};
\node[below = 12pt of v3.west, anchor=west] (v5){$5:x_6\succ x_3 \succ$ \fbox{$x_5$}\,,};
\node[right = -1pt of v5] (v6){$6:x_4\succ x_3 \succ$ \fbox{$x_6$}\,,};
\end{tikzpicture}
\end{center}

\noindent Showing that \RO{} is NP-hard when the underlying graph is a tree and presenting a simple polynomial-time algorithm for a special restricted case on a path (where the given agent is an endpoint of the path), the authors explicitly leave as open questions the general case on paths and other special cases (including restricted input preferences).

In this work, we answer several open questions~\cite{GouLesWil2017,SW18} and draw a comprehensive picture of the computational complexity of \RO.
\begin{compactitem}[$\bullet$]
\item Our main contribution is a polynomial-time algorithm on paths (\Cref{thm:path}).
  This algorithm combines a multitude of structural observations in a nontrivial way and requires a sophisticated analysis.
  \item Second, we show NP-hardness on complete graphs even if all preference lists have length at most four~(\Cref{thm:NP-c-length-4}).
  We complement this hardness by giving a linear-time algorithm for preferences lists of length at most three~(\Cref{thm:preflength-three}).
\item Moreover, we prove NP-hardness for generalized caterpillars (\Cref{thm:cater}) and thereby narrow the gap between tractable and intractable cases of the problem.
  \end{compactitem}
The NP-hardness from \Cref{thm:NP-c-length-4} implies that the problem is already NP-hard even if the agents are allowed to swap without restrictions and no agent has more than three objects which she prefers to her initial one.
The hardness reduction can be adapted to also show NP-hardness for the case where the maximum vertex degree of the graph is five and the preference lists have length at most four.


\paragraph{Related Work.}

\citet{GouLesWil2017} proposed the model of distributed object reallocation via swaps on social networks, and showed that \RO{} is NP-hard on trees. 
Moreover, they showed polynomial-time solvability on stars and for a special case on paths, namely when testing whether an object is reachable for an agent positioned on an endpoint of the path. They also indicated that the problem is polynomial-time solvable on paths when the agent and the object are at constant distance.
Notably, they explicitly asked for a polynomial-time algorithm on paths in general and
describe the problem as being at the frontier of tractability, despite its simplicity.
Besides \RO{}, \citet{GouLesWil2017} also considered the questions of reachability of a particular allocation, the \textsc{Reachable Assignment} problem, and existence of a Pareto-efficient allocation; both are shown to be NP-hard.

\citet{SW18} studied the parameterized complexity of \RO{} with respect to parameters such as the maximum vertex degree of the underlying graph or the overall number of swaps allowed in a sequence.
They showed several parameterized intractability results and also fixed-parameter tractable cases (none of which covers our results).
Notably, in their conclusion they suggested to study restrictions on the preferences~(as we do in this paper). 
 
Other examples of recently studied problems regarding allocations of indivisible resources under social network constraints are envy-free allocations~\cite{BKN18,BCGLMW18}, Pareto-optimal allocations~\cite{IP18}, and two-sided stable matching~\cite{ArcVas2009,AnsBhaHoe2017}. See the work of \citet{GouLesWil2017} fore more related work.

\subsection{Preliminaries}\label{sec:prelim}

Let $V = \{1,2,\ldots, n\}$ be a set of $n$ agents and $X\!=\!\{x_1,x_2$, $\ldots,x_n\}$ be a set of $n$ objects.
Each agent~$i \in V$ has a \myemph{preference list over a subset~$X_i\subseteq X$} of the objects, which is a strict linear order on $X_i$.
This list is denoted as $\succ_i$; we omit the subscript if the agent is clear from the context.
For two objects~$x_j, x_{j'}\in X_i$, the notation~$x_j \succ_i x_{j'}$ means that agent~$i$ \myemph{prefers~$x_j$ to~$x_{j'}$}.
A \myemph{preference profile~$\mathcal{P}$ for the agent set~$V$} is a collection~$(\succ_i)_{i\in V}$ of preference lists of the agents in $V$.
An \myemph{assignment} is a bijection~$\sigma\colon V \to X$, where each agent~$i$ is assigned an object~$\sigma(i)\in X_i$.
We say that \myemph{an assignment~$\sigma$ admits a rational trade for two agents~$i, i'\in V$} if 
$\sigma(j) \succ_i \sigma(i)$ and $\sigma(i) \succ_j \sigma(j)$.

We assume that the agents from $V$ form a social network such that pairs of adjacent agents can trade their objects. 
The social network is modeled by an undirected graph~$G=(V,E)$ with $V$ being also the vertex set and $E$ being a set of edges on $V$. 
We say that an assignment~$\sigma$ admits \myemph{a swap for two agents~$i$ and $i'$}, denoted as $\tau\!=\!\{\{i,\sigma(i)\},\{i',\sigma(i')\}\}$,
if it admits a rational trade for $i$ and $i'$ and 
the vertices corresponding to~$i$ and $i'$ are adjacent in the graph, i.e., $\{i,i'\}\in E$.
Accordingly, we also say that objects $\sigma(i)$ (resp.\ $\sigma(i')$) \myemph{passes through} edge~$\{i,i'\}$.
By definition, an object can pass through an edge at most once.

A \myemph{sequence of swaps} is a sequence~$(\sigma_0,\sigma_1,\ldots, \sigma_t)$ of assignments where for each index~$k \in \{0,1,\ldots, t-1\}$ 
there are two agents~$i,i' \in V$ for which $\sigma_k$ admits a swap 
such that  
\begin{inparaenum}[(1)]
\item $\sigma_{k+1}(i)=\sigma_{k}(i')$,
\item $\sigma_{k+1}(i')=\sigma_{k}(i)$, and 
\item for each remaining agent~$z\in V\setminus \{i,i'\}$ it holds that $\sigma_{k+1}(z)=\sigma_{k}(z)$.
\end{inparaenum}
We call an \myemph{assignment~$\sigma'$ reachable from another assignment~$\sigma$} if there is a sequence~$(\sigma_0,\sigma_1,\ldots, \sigma_t)$ of swaps such that $\sigma_0=\sigma$ and $\sigma_t=\sigma'$.
The reachability relation defines a partial order on the set of all possible assignments.
Given an initial assignment~$\sigma_0$, we say that \myemph{an object~$x\in X$ is reachable for an agent~$i$} if there is an assignment~$\sigma$ which is reachable from $\sigma_0$ with $\sigma(i)=x$.


We study the following computational problem~\cite{GouLesWil2017}, called \RO{} (\ROs),
which has as input an agent set~$V$, an object set~$X$, a preference profile for $V$, an undirected graph~$G=(V,E)$ on $V$, 
an initial assignment~$\sigma_0$,
an agent~$I\in V$ and an object~$x\in X$,
and ask whether $x$ is reachable for $I$ from $\sigma_0$.
\noindent Note that \ROs{} is contained in NP~\cite{GouLesWil2017}.

\newcommand{\inNP}{\ROs{} is in NP.}
\begin{proposition}\label[proposition]{prop:RO-in-NP}
  \inNP
\end{proposition}

\begin{proof}
  Since each object can pass through each edge at most once
  it follows that each reachable assignment uses a different edge to swap objects.
  Hence, a certificate for a \ROs{} instance with $n$~agents is a sequence of swaps where each object appears at most $O(n^2)$ times, i.e. the length of the sequence is $O(n^3)$.
\end{proof}

We consider simple undirected graphs~$G=(V,E)$ containing vertices~$V$ and edges~$E\subseteq\binom{V}{2}$. We assume the reader to be familiar with basic graph classes such as paths, cycles, trees, and complete graphs (cliques)~\cite{Diestel2012}.
A \myemph{caterpillar} is a tree such that removing all leaves yields a path (i.e, all vertices are within distance at most one of a central path). We call the edges to the leaves attached to the central path \myemph{hairs} of \myemph{length}~1. A \myemph{generalized caterpillar} has hairs of length~$h\ge 1$. 

\section{Preferences of Length at Most Three}\label{sec:lenth=3}

In this section, we provide a linear-time algorithm for \ROs{} in the special case of preferences with length at most three. The main idea is to reduce \ROs{} to reachability in directed graphs. The approach is described in \cref{alg:length=3}.
Take the example from the introduction,
and delete object~$x_1$ from agent~$3$'s preference list
and object~$x_3$ from agent~$2$' preference list.
Then, our algorithm finds the following swap sequence for object~$x_3$ to reach agent~$1$:
$4\leftrightarrow 3$, $3\leftrightarrow 2$, $2\leftrightarrow 1$, $4\leftrightarrow 5$, $5\leftrightarrow 6$, $6\leftrightarrow 1$; here ``$i \leftrightarrow j$'' means that agent~$i$ swaps with agent~$j$ over the objects currently held by them.

  Throughout this section, we assume that each agent~$i \in \{1,\ldots, n\}$, when initially holding object~$x_i$, i.e., $\sigma_0(i)=x_i$, has either two or three objects in her preference list; we can ignore those agents who only have one object in their respective preference lists because they will never trade with someone else.
  We aim to determine whether object~$x_n$, which is initially held by agent~$n$, is reachable for agent~$1$.
  To ease the reasoning, we define an equivalent and succinct notion of swap sequences, which only focuses on the relevant swaps.
  Let $\tau=\{\{i,x\},\{j,y\}\}$ be a swap for which an assignment~$\sigma$ admit.
  Then, we use the notation \myemph{$\sigma/\tau$} to denote the assignment
  that results from $\sigma$ by performing the swap~$\tau$.
  In an overloading manner, we say that a \emph{sequence~$\phi = (\tau_1,\tau_2,\ldots,\tau_m)$ of swaps} is a \myemph{valid swap sequence for some start assignment~$\sigma_0$}
  if there exists a sequence of swaps~$(\sigma_0,\sigma_1,\ldots,\sigma_m)$~(see \cref{sec:prelim}) such that for each $z\in \{1,\ldots,m\}$, it holds that $\sigma_z=\sigma_{z-1}/\tau_z$.

  We first observe a property which 
  allows us to exclusively focus on a specific valid swap sequence where each swap involves swapping object~$x_q$, if object~$x_q$ reaches agent~$1$ during the swap sequence.
  \newcommand{\swappq}{
    Let $\phi=(\tau_1,\tau_2,\ldots,\tau_m)$ be a valid swap sequence for $\sigma_0$ such that the agent~$1$ obtains~$x_n$ in the last swap.
    Consider two objects~$x_p$ and $x_q$.
    If $\phi$ contains a swap~$\tau_{r}$ with $\tau_r=\{\{1,x_p\},\{k,x_q\}\}$ (for some agent~$k$),
    then let $\phi'=(\tau'_1,\tau'_2,\ldots,\tau'_s)$ be a subsequence of the prefix~$\phi_0=(\tau_1,\tau_2,\ldots,\tau_r)$ of $\phi$, up to~(including) $\tau_r$,
    that consists of exactly those swaps~$\tau$ from $\phi_0$ with $\{k_\tau,x_q\}\in \tau$ for some agent~$k_\tau$,
    Define $a_{s+1}=p$ and $a_z$, $1\le z \le s$, such that $\{a_z,x_q\}\in \tau'_z$.
    The following holds.
   \begin{compactenum}[(i)]
     \item \label{swap-p-q:last-swap}
     $\tau'_s=\tau_r$.
     \item\label{swap-p-q:a1} $a_1=q$.

     \item\label{swap-p-q:prefer} For each $z\in \{1,\ldots, s\}$ agent~$a_z$ prefers~$x_{a_{z+1}}$ to~$x_{q}$.
     
      \item\label{swap-p-q:preference-list} For each~$z\in \{2,\ldots, s\}$ agent~$a_{z}$ has preference list $x_{a_{z+1}} \succ x_{q} \succ$\fbox{$x_{a_z}$}.
      

      \item \label{swap-p-q:swap-p-q} $\phi'$ is a valid swap sequence for $\sigma'_0$, where $\sigma'_0(i)=\sigma_0(i)$ for all $i\notin\{1,p\}$ while $\sigma_0'(1)=x_p$ and $\sigma_0'(p)=x_1$.
    
    \item\label{swap-p-q:no-xn} If agent~$1$ prefers $x_n$ to $x_q$, then
    no agent~$a_z$, $3\le z\le s$, contains object~$x_n$ in her preference list.
  \end{compactenum}
  }
  \begin{lemma} \label[lemma]{lem:swap-p-q}
    \swappq
  \end{lemma}

    \begin{proof}
    Statement~\eqref{swap-p-q:last-swap}: By definition, the last swap~$\tau_r$ in $\phi_0$ contains $\{j,x_q\}$ for some agent~$j$, since $\phi'$ contains exactly those swaps from $\phi_0$ that involves swapping object~$x_q$, it follows that the last swap in $\phi'$ must be $\tau_r$.

    Statement~\eqref{swap-p-q:a1} holds because agent~$q$ initially holds object~$x_q$ and thus, in order to make object~$x_q$ reach agent~$1$, agent~$q$ is the first agent to give away object~$x_q$.

   Now, we turn to statement~\eqref{swap-p-q:prefer}.
   Note that all agents~$a_z$ are distinct because by the definition of rational trades no agent will give away the same object~$x_q$ more than once. 
    We show statement~\eqref{swap-p-q:prefer} by induction on $z$, $1\le z\le s$, starting with $z=s$.

    First of all, by the first statement, 
    we know that $\tau'_s=\tau_r$,
    i.e, in this swap~$\tau'_s$ agent~$a_{s}$ swaps with agent~$1$ over objects~$x_q$ and $x_p=x_{a_{s+1}}$.
    Since $\tau_s'$ is also a rational trade, this means that agent~$a_{s}$ must prefer object~$x_{a_{s+1}}$ to object~$x_q$.

    For the induction assumption, let us assume that for each~$i \ge z$, agent~$a_i$ prefers $x_{a_{i+1}}$ to $x_q$.
    Now, we consider agent~$a_{z-1}$, and we aim to show that $a_{z-1}$ prefers $x_{a_{z}}$ to $x_q$.
    By definition, we know that
     $\tau'_{z-1}=\{\{a_{z-1},x_q\},\{a_z,y\}\}$ for some object~$y$,
     i.e., agent~$a_{z-1}$ gives object~$x_q$ to agent~$a_{z}$ in order to obtain another object~$y$.
    Thus, agent~$a_{z-1}$ must prefer~$y$ to~$x_q$ and agent~$a_z$ must prefer ~$x_q$ to~$y$.
    Since each agent has at most three objects in her preference list and since by our induction assumption we know that $a_z$ already prefers $x_{a_{z+1}}$ to $x_q$
    we infer that that $y$ is the initial object of agent~$a_z$, that is, $y=x_{a_z}$.

    Next, to show statement~\eqref{swap-p-q:preference-list}, let us re-examine the preferences of agents~$a_z$, $1\le z \le s$ implied by statement~\eqref{swap-p-q:prefer}.
    Since each $a_z$, $1\le z\le s$, prefers $x_{a_{z+1}}$ to $x_q$,
    the preference list of each agent~$a_{z}$, $2\le z\le s$, is 
    $x_{a_{z+1}} \succ x_q \succ$ \fbox{$x_{a_z}$}; again, recall that each agent has at most three objects in her preference list.

    Now, we show that $\phi'$ is a valid swap sequence for $\sigma'_0$, i.e.,
    there exists a swap sequence~$(\rho_0,\rho_1,\rho_2,\ldots,\rho_s)$, a sequence of assignments, such
    that for $\rho_0=\sigma'_0$ and for each $z\in \{1,\ldots, s\}$ it holds that $\rho_{z}=\rho_{z-1}/\tau'_z$.   
    We prove this by showing the following 
    \begin{compactenum}[(1)]
      \item $\sigma'_0$ is an assignment and admits swap~$\tau'_1$.
      \item  for each $z\in \{1,\ldots, s-1\}$ if $\sigma'_{z-1}$ is an assignment 
      and admits swap~$\tau'_{z}$,
      then $\sigma'_{z}\coloneqq \sigma'_{z-1}/\tau'_{z}$ is also an assignment 
      and admits swap~$\tau'_{z+1}$,
      \item $\sigma'_{s}\coloneqq \sigma'_{s-1}/\tau'_{s}$ is an assignment.
    \end{compactenum}

    Clearly,~$\sigma'_0$ is an assignment and admits swap~$\tau'_1$, which is $\tau'_1=\{\{a_1,x_q\}, \{a_2,x_{a_2}\}\}$.
    This means that $\sigma'_1$, defined as $\sigma'_1\coloneqq \sigma'_0/\tau'_1$,
    is an assignment.
    Now, consider~$z\in\{1,\ldots,s-2\}$, and assume that $\sigma'_{z-1}$ is an assignment and admits swap~$\tau_{z}$, which is $\tau_z=\{\{a_z,x_q\}, \{a_{z+1},x_{a_{z+1}}\}\}$.
    Thus,~$\sigma'_{z}\coloneqq \sigma'_{z-1}/\tau'_{z}$ is an assignment.
    By the definitions of $T\coloneqq \{\tau'_1,\ldots, \tau'_{z}\}$, we infer that $\sigma'_z(a_{z+1})=x_q$ and $\sigma'_{z}(a_{z+2})=x_{a_{z+2}}$; the latter holds because no swap from $T$ has involved agent~$a_{z+2}$.
    Thus, by the preference lists of $a_{z+1}$ and $a_{z+2}$ it holds that $\sigma'_{z}$ admits swap~$\tau'_{z+1}$, which is $\tau'_{z+1}=\{\{a_{z+1},x_q\},\{a_{z+2}, x_{a_{z+2}}\}\}$.
    Finally, we obtain that $\sigma'_{s-1}\coloneqq \sigma'_{s-2}/\tau_{s-1}$ is an assignment
    such that $\sigma'_{s-1}(a_{s-1}) =x_q$ and $\sigma'_{s-1}(1)=x_{a_{s+1}}=x_p$.
    By the preference list of agent~$1$ and $a_{s-1}$, it is clear that $\sigma'_{s-1}$ admits swap~$\tau'_s$.
    Define $\sigma'_s=\sigma'_{s-1}/\tau_s$ and it is clear it is an assignment.
    Concluding,~$\phi'$ is a valid swap sequence for $\sigma'_0$. 

    Finally, we show statement~\eqref{swap-p-q:no-xn}.
    Assume that agent~$1$ prefers object~$x_n$ to object~$x_q$, implying that $x_n\neq x_q$.
    We assume that $s-1\ge 4$ as otherwise the set~$\{3,\ldots,s-1\}$ is empty and we are done with the statement.
    Suppose, for the sake of contradiction, that some agent~$a_z$ with $3\le z\le s-1$ has object~$x_n$ in her preference list.
    By the preference list of $a_z$ we infer that the object~$x_n$ is either $x_{a_z}$ or $x_{a_{z+1}}$ because $x_n\neq x_q$.
    If $x_n=x_{a_z}$, then after the swap~$\tau'_z$, agent~$a_{z-1}$ will obtain $x_n$, which is her most preferred object---a contradiction to agent~$1$ receiving object~$x_n$ after $\phi'$.
    If $x_n=x_{a_{z+1}}$, then after the swap~$\tau'_{z+1}$, agent~$a_{z}$ will obtain $x_n$, which is her most preferred object---a contradiction to agent~$1$ receiving object~$x_n$ after $\phi'$.
\end{proof}

Now, we are ready to give the main algorithm for the case where the length of the preference list of each agent is bounded by three based on solving reachability in a directed graph (\cref{alg:length=3}).

\renewcommand\ArgSty{\normalfont}

\begin{algorithm}[t]
  \DontPrintSemicolon
  \caption{Algorithm for \ROs{} with preference list length at most three.}
  \label[algorithm]{alg:length=3}
  \small
  \SetKwInOut{Input}{Input}
  \Input{\small Agent set~$V$ with preference lists~$(\succ_i)_{i\in V}$ over object set~$X$, and the underlying graph~$(V,E)$}

  $D\!\coloneqq\! (V,F)$ with $F\!\coloneqq\! \{(i,j)\! \mid\! \{i,j\} \in E \wedge x_j \succ_i x_n \wedge x_n \succ_j x_j\}$.\label{line:D}

  \label{line:n-1}  \If{$D$ admits a directed path~$P$ from $n$ to $1$}
  {\label{line:path-n-1}  \textbf{return} yes}

    \lIf{$x_n\succ_1 x_w \succ_1 $ \fbox{$x_1$} for some~$x_w\neq x_n$}
    {
      $D_1\coloneqq (V,F_1)$ with $F_1\coloneqq \{(i,j) \mid \{i,j\} \in E \wedge x_j \succ_i x_w \wedge x_w \succ_j x_j\}$.\label{line:D1}


     \ForEach{Directed path~$P_1=(w,n,a_1,\ldots,a_s,1)$ from $w$ to $1$ in $D_1$ such that the first arc on~$P_1$ is $(w,n)$\label{line:D1-wn}}{
       $D_2\coloneqq D-\{a_1,\ldots,a_s\}+\{(j,1)\mid x_w \succ_j x_n\}$
       \label{line:D2}

        \lIf{$D_2$ admits a directed path~$P_2$ from $n$ to $1$ such that the first arc on $P_2$ is $(n,w)$}
        {\textbf{return} yes \label{line:1-w-n}}
      }
      \ForEach{Directed path~$P_1=(w,a_1,\ldots,a_s,1)$ from $w$ to $1$ in $D_1$ such that the first arc on~$P_1$ is \emph{not} $(w,n)$, i.e., $a_1\neq n$}
      {
       $D_3\coloneqq D-\{w,a_1,\ldots,a_s\}+\{(j,1)\mid x_w \succ_j x_n\}$
        \label{line:D3}

        \lIf{$D_3$ admits a directed path~$P_3$ from $n$ to $1$ such that the first arc on $P_3$ is \emph{not} $(n,w)$}
        {\textbf{return} yes        \label{line:1-no-w-n}}    }
  }
  \textbf{return} no
\end{algorithm}

\newcommand{\lengththree}{%
\ROs{} for preference list length at most three can be solved in~$O(n+m)$~time, where~$m$ is the number of edges in the underlying graph.%
}
\begin{theorem}\label{thm:preflength-three}
  \lengththree
\end{theorem}

  \begin{proof}[Proof sketch.]
  We claim that \cref{alg:length=3} solves our problem in linear time.
  Assume that object~$x_n$ is reachable for agent~$1$, i.e., there exists a valid swap sequence~$\phi=(\tau_1,\tau_2,\ldots,\tau_m)$ for $\sigma_0$ such that $\tau_m=\{\{1,x\}, \{j,x_n\}\}$ for some object~$x$ which exists in the preference list of $1$ and some agent~$j$ that has $x_n$ in her preference list.
  There are two cases for $x$: either $x=x_1$ or $x\neq x_1$.
  If $x=x_1$, then 
  using $x_p=x_1$ and $x_q=x_n$, the sequence~$\phi'$ as defined in~\cref{lem:swap-p-q}
  is a valid swap sequence for $\sigma'_0$ with $\sigma'_0=\sigma_0$.
  By the properties in \cref{lem:swap-p-q}\eqref{swap-p-q:preference-list}, graph~$D$ as constructed in Line~\ref{line:D} must contain a path from $n$ to $1$.
  Thus, by Line~\ref{line:path-n-1} our algorithm returns yes.

  If $x\neq x_1$, implying that the preference list of agent~$1$ is $x_n \succ x_w \succ $ \fbox{$x_1$} for some object~$x_w$ such that $x=x_w$,
  then $\phi$ has a swap~$\tau_r$ such that $\tau_r=\{\{1,x_1\}, \{k,x_w\}\}$ for some agent~$k$.
  By \cref{lem:swap-p-q} (using $x_p=x_1$ and $x_q=x_w$),
  the sequence~$\phi'=(\tau'_1,\ldots,\tau'_s)$ as defined in \cref{lem:swap-p-q} is a valid swap sequence for $\sigma'_0$ with $\sigma'_0=\sigma_0$.
  Let $a_z$, $1\le z \le s$, be the agents such that $\{a_z,x_w\} \in \tau'_z$.
  By \cref{lem:swap-p-q}\eqref{swap-p-q:a1}, we know that $\tau'_1=\{\{w,x_w\},\{a_2,x_{a_2}\}\} = \{\{a_1,x_{a_1}\},\{a_2,x_{a_2}\}\}$.
  
  There are two cases for $a_2$: either $a_2=n$ or $a_2\neq n$.
  If $a_2=n$, then by the properties in \cref{lem:swap-p-q}\eqref{swap-p-q:preference-list},
  it follows that $\phi'$ defines a directed path~$(a_1,a_2),\ldots,(a_{s-1},a_s),(a_s,1)$ in $D_1$ (Line~\ref{line:D1}) with $(a_1,a_2)=(w,n)$.
  By $\sigma'_0$ and \cref{lem:swap-p-q}\eqref{swap-p-q:preference-list}, we have that $(a_1,a_2),(a_2,a_3), \ldots, (a_{s-1},a_s), (a_s,1)$ is a directed path in graph~$D_1$ as defined in Line~\ref{line:D1} such that $(a_1,a_2) = (w,n)$ so that the if condition from Line \ref{line:D1-wn} holds.
  By \cref{lem:swap-p-q}\eqref{swap-p-q:no-xn}, no agent from $\{a_3,a_4,\ldots,a_s\}$ is involved in a swap from $\phi$ which includes $x_n$. 
  
  By the above properties of $\tau_m$ and $\phi$, it follows that using $x_p=x_w$ and $x_q=x_n$,
  the sequence~$\phi''=(\rho_1,\rho_2,\ldots,\rho_{t})$ as defined in \cref{lem:swap-p-q} is a valid swap sequence for 
  $\sigma''_0$ with $\sigma''_0(1)=x_w$, $\sigma''_0(w)=x_1$, and for all $i\notin \{1,w\}$, $\sigma''_0(i)=\sigma_0(i)$.
  Recall that $\phi'$ contains exactly those swaps from $\phi_0$ which involve swapping object~$x_w$.
  Thus, it must hold that $\rho_1=\tau'_1$.
  Let $b_z$, $1\le z\le t$, be the agents such that $\{b_z,x_n\}\in \rho_z$.
  Then, $b_1=n$ and $b_2=w$ and no agent from $\{b_3,\ldots,b_t\}$ is from $\{a_3,\ldots,a_s\}$.
  Thus, by $\sigma''_0$ and \cref{lem:swap-p-q}\eqref{swap-p-q:preference-list}, we have that $(b_1,b_2),(b_2,b_3), \ldots, (b_{t-1},b_t), (b_t,1)$ is a directed path in graph~$D_2$ as defined in Line~\ref{line:D2} such that $(b_1,b_2)=(n,w)$.
  Indeed, by Line \ref{line:1-w-n}, our algorithm returns yes.
  
  Now, we turn to the other case, namely, $a_2\neq n$.
  Again, by the properties of $\tau_m$ and $\phi$, it follows that, using $x_p=x_w$ and $x_q=x_n$,
  the sequence~$\phi''=(\rho_1,\rho_2,\ldots,\rho_{t})$ as defined in \cref{lem:swap-p-q} is a valid swap sequence for 
  $\sigma''_0$ with $\sigma''_0(1)=x_w$, $\sigma''_0(w)=x_1$, and for all $i\notin \{1,w\}$, $\sigma''_0(i)=\sigma_0(i)$.
  Let $b_z$, $1\le z\le t$, be the agents such that $\{b_z,x_n\}\in \rho_z$.
  
  We aim to show that no swap from $\phi$ that involves swapping object~$x_n$ will involve agent~$a_z$, $1\le z\le s$, i.e.,
  $\{a_1,\ldots,a_s\}\cap \{b_1,\ldots,b_t\}=\emptyset$.
  By $\phi'_0$ and the properties in \cref{lem:swap-p-q}\eqref{swap-p-q:preference-list},
  it follows that $\phi'$ defines a directed path~$(a_1,a_2),\ldots,(a_{s-1},a_s),(a_s,1)$ in $D_1$ (Line~\ref{line:D1}) such that $a_2 \neq n$.
  By \cref{lem:swap-p-q}\eqref{swap-p-q:preference-list} and \cref{lem:swap-p-q}\eqref{swap-p-q:no-xn}, no agent from $\{a_2,a_3,a_4,\ldots,a_s\}$ is
  involved in any swap from $\phi$ which includes object~$x_n$. 
  Neither will agent~$a_1$ be involved in any swap from $\phi$ which includes object~$x_n$ because of the following.
  After swap~$\tau'_1$ agent~$a_1$ obtains object~$x_{a_2}$ which is not $x_n$, so if she would be involved in swapping $x_n$, then she would obtain $x_n$ as her most preferred object and never give away $x_n$.
  Thus, indeed we have   $\{a_1,\ldots,a_s\}\cap \{b_1,\ldots,b_t\}=\emptyset$.
  By $\sigma''_0$  and \cref{lem:swap-p-q}\eqref{swap-p-q:preference-list}, we have that $(b_1,b_2),(b_2,b_3), \ldots, (b_{t-1},b_t), (b_t,1)$ is a directed path in graph~$D_3$ as defined in Line~\ref{line:D3} such that $(b_1,b_2)\neq (n,w)$.
  Indeed, by Line \ref{line:1-no-w-n}, our algorithm returns yes.

  For the converse direction, if our algorithm return yes, then one of the three lines \ref{line:path-n-1}, \ref{line:1-w-n}, \ref{line:1-no-w-n}, returns yes.
  It is not hard to check that the corresponding constructed path(s) indeed define(s) a desired valid swap sequence.
  As for the running time, our algorithm constructed at most four directed graphs~$D, D_1, D_2$, and~$D_3$, each of them with $O(n)$~arcs.
  We only show how to construct $D$ in $O(n+m)$~time, the construction for $D_1$ is analogous and~graphs~$D_2$ and $D_3$ are derived from $D$.

  To construct $D=(V,F)$, we go through each edge~$\{i,j\}$ with $i\neq n$ and $j \neq n$ in the underlying graph and do the following in $O(1)$ time; we tackle agent~$n$ separately.
  \begin{compactitem}
    \item Check whether the preference list of~$i$ is $x_j \succ x_n \succ $\fbox{$x_i$} and~$j$ prefers~$x_n$ over~$x_j$. If so, we add the edge~$(i,j)$.
    \item Check whether the preference list of~$j$ is $x_i \succ x_n \succ $\fbox{$x_j$} and~$i$ prefers~$x_n$ over~$x_i$. If so, we add the edge~$(j,i)$.
  \end{compactitem}
  Now, we consider agent~$n$.
  For each object~$x_j$ (there are at most two of them) that agent~$n$ prefers to her initial object~$x_n$, we do the following in $O(m)$~time.
  \begin{compactitem}
    \item Check whether $\{n,j\} \in E$.
    \item Check whether $j$ prefers $x_n$ to $x_j$.
  \end{compactitem}
  Add the arc~$(n,j)$ to~$D$ only if the above two checks are positive.
  In this way, each vertex from $V\setminus \{n\}$ has at most two in-arcs and at most one out-arc and vertex~$n$ has at most two out-arcs but no in-arcs.
  Thus, $|F| \in O(n)$.

  For each of the three graphs, each with $O(n)$ arcs, checking whether the specific path stated in the corresponding line exists can be done in $O(n)$~time. 
  In total, the algorithm runs in $O(n+m)$~time.
\end{proof}

\section{Paths}\label{sec:paths}

In this section we prove that \ROs{} on paths is solvable in~$O(n^4)$ time.
This answers an open question by \citet{GouLesWil2017}.
The proof consists of two phases.
In the first phase, we analyze sequences of swaps in paths and observe a crucial structure regarding possible edges along which a pair of objects can be traded.
In the second phase, we use this structure to reduce \ROs{} on paths to \textsc{2-SAT}.

Let~$P = (\{1,2,\ldots, n\},\{\{i,i+1\}\mid 1\le i < n\})$ be a path, and w.l.o.g.\ let~$\sigma_0(n) = x$ (see \Cref{fig:exampleInput}).
Note that if there are some agents~$n+1,\ldots$, then their initially assigned objects can never be part of any swap that is necessary for~$x$ to reach~$I$.
If such a swap was necessary, then there would be a swap that moves~$x$ away from~$I$, that is, to an agent with a higher index, and therefore the agent that gave away~$x$ now possesses an object that she prefers over~$x$ and she will never accept~$x$ again.
Thus, object~$x$ could never reach~$I$.

In the following,``an object~$y$ moves to the right''  means that object~$y$ is given to an agent with a higher index than the agent that currently holds it.
An object~$z$ is ``to the left (resp.\ right) of some other object~$a$'' when
the agent that initially holds object~$z$ has a smaller index than the agent that initially holds object~$a$.

\begin{figure}
  \centering
\begin{tikzpicture}[scale=0.8, every node/.style={scale=0.9}]
\node[circle,draw, label=below:1] at (0,0) (v1) {};
\node[circle,draw, label=below:2] at (1,0) (v2) {};
\node at (2,0) (v3) {$\cdots$};
\node[circle,draw, label=below:$I$] at (3,0) (v4) {};
\node[circle,draw, label=below:$I+1$] at (4,0) (v5) {};
\node[circle,draw, label=below:$I+2$] at (5,0) (v6) {};
\node at (6,0) (v7) {$\cdots$};
\node[circle,draw, label=below:$n-1$] at (7,0) (v8) {};
\node[circle,draw, label=below:$n$, label=above:$x$] at (8.5,0) (v9) {};

\draw (v1) -- (v2);
\draw (v2) -- (v3);
\draw (v3) -- (v4);
\draw (v4) -- node[above] {1} (v5);
\draw (v5) -- node[above] {2} (v6);
\draw (v6) -- node[above] {$\cdots$} (v7);
\draw (v7) -- (v8);
\draw (v8) -- node[above] {$n - I$} (v9);
\end{tikzpicture}
\caption{An example of a path where the name of the agent is denoted below its vertex.
  The target object~$x$ is the initial object of~$n$
  and is depicted above the corresponding vertex~$n$.
  The edges to the ``right'' of~$I$ are enumerated in order to define types of objects.}
\label{fig:exampleInput}
\end{figure}
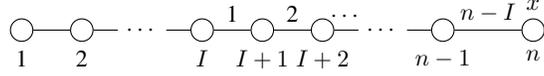

We start with a helpful lemma that will be used multiple times.
It states that for each pair of objects there is at most one edge along which these two objects can be swapped in order to pass each other.

\newcommand{\swappos}{
For each agent~$i$ and each two distinct objects object~$w$ and $y$
there exists at most one edge, denoted as $\{j,j+1\}$, 
such that for each sequence~$\xi$ of swaps
the following statement holds:
If after~$\xi$ agent~$i$ holds object~$w$, and if objects~$w$ and~$y$ are swapped in~$\xi$, then objects~$w$ and~$y$ are swapped over the edge~$\{j,j+1\}$.
Deciding whether such an edge exists and computing it if it exists takes~$O(n)$ time.

}
\begin{lemma}
  \label{lem:swapPos}
  \swappos
\end{lemma}

\begin{proof}
Let~$a$ be the agent that initially holds~$y$ and let~$c$ be the agent that initially holds~$w$.
Assume without loss of generality that $a < c$.
By the definition of rational trades, no agent takes an object back that she gave away before.
Hence, if~$w$ and~$y$ are swapped at some point, then they need to travel ``towards and never away from each other'' in the path before the swap occurs.
Thus, we only need to consider the set of agents~$M = \{d \mid a \leq d \leq c\}$ as potential candidates for~$j$ and~$j+1$.
We first find an agent~$e$ that has to get object~$y$ before she can get object~$w$.
We consider two cases: Either~$i \in M$ or~$i \notin M$ (see \Cref{fig:casedistinction}).

\noindent \textbf{Case~1: $i\notin M$.}~That is, $i < a$ because $w$ and $y$ are swapped before $w$ reaches $i$, and~$M$ contains all agents on a connected subpath including agent~$c$.
Then,~$w$ needs to pass agent~$a$ at some point in order for~$w$ to reach~$i$. Let $e=a$. 

\noindent \textbf{Case~2: $i\in M$.} That is, agent~$i$ is initially between objects~$w$ and~$y$.
Since we are only interested in sequences of swaps until agent~$i$ gets object~$w$, we can assume w.l.o.g.\ that~$j \geq i$.
We assume further that agent~$i$ prefers~$w$ over~$y$ as she would otherwise not accept~$w$ after giving away~$y$.
Thus, agent~$i$ is the agent~$e$ we are looking for.

Next we will show that there is only one edge to the right of agent~$e$ where~$w$ and~$y$ can be swapped.
Consider the agent~$b \in M$ such that all agents in the set~$M_b = \{d \mid e < d < b\}$ prefer~$w$ over~$y$ and~$b$ prefers~$y$ over~$w$.
Note that~$b$ cannot obtain object~$y$ before object~$w$ as she would not accept~$w$ at any future point in time and therefore~$w$ will never pass~$b$.
Thus, it holds that~$j+1 \leq b$.
Any agent in~$M_b$ prefers~$w$ over~$y$ and therefore would not agree on swapping~$w$ for~$y$.
Thus, it holds that~$j+1 \geq b$ and therefore~$j+1 = b$.

Observe that the agent~$b$ can be computed in linear time and is always the only candidate for being agent~$j+1$ we are looking for.
\end{proof}

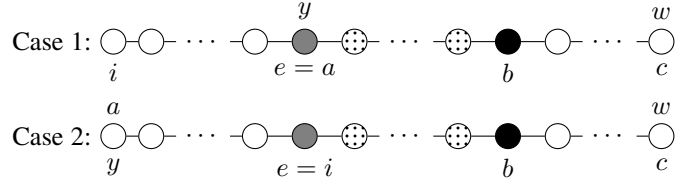
\begin{figure}
\centering
  \def \dist {1ex}
\begin{tikzpicture}
\node[draw,circle, label=below:{$i$}, fill=white] at (0,0) (i) {};
\node[left = 0pt of i] {Case~1:~};
\node[draw,circle, right = \dist of i]  (dummy4) {};
\node[right = \dist of dummy4] (dummy3) {$\cdots$};
\node[draw,circle, right = \dist of dummy3] (dummy2) {};
\node[draw,circle, label=below:{$e=a$}, label=above:$y$, fill=black!50!white, right = \dist*2 of dummy2]  (v1) {};
\node[draw,circle, pattern=dots, right = \dist*2 of v1]  (v2) {};
\node[right = \dist of v2]  (v3) {$\cdots$};
\node[draw,circle, pattern=dots, right = \dist of v3]  (v4) {};
\node[draw,circle, label=below:{$b$}, fill=black, right = \dist*2 of v4] (v5) {};
\node[draw,circle, fill=white, right = \dist*2 of v5]  (v6) {};
\node[right = \dist of v6] (v7) {$\cdots$};
\node[draw,circle, label=below:$c$, label=above:$w$, right = \dist of v7] (v8) {};

\draw (dummy4) -- (i);
\draw (dummy3) -- (dummy4);
\draw (dummy3) -- (dummy2);
\draw (dummy2) -- (v1);
\draw (v1) -- (v2);
\draw (v2) -- (v3);
\draw (v3) -- (v4);
\draw (v4) -- (v5);
\draw (v5) -- (v6);
\draw (v6) -- (v7);
\draw (v7) -- (v8);
\end{tikzpicture}

\begin{tikzpicture}
\node[draw,circle, label=below:{$y$}, label=above:$a$, fill=white] at (0,0) (i) {};
\node[left = 0pt of i] {Case~2:~};
\node[draw,circle, right = \dist of i]  (dummy4) {};
\node[right = \dist of dummy4] (dummy3) {$\cdots$};
\node[draw,circle, right = \dist of dummy3] (dummy2) {};
\node[draw,circle, label=below:{$e=i$},  fill=black!50!white, right = \dist*2 of dummy2]  (v1) {};
\node[draw,circle, pattern=dots, right = \dist*2 of v1]  (v2) {};
\node[right = \dist of v2]  (v3) {$\cdots$};
\node[draw,circle, pattern=dots, right = \dist of v3]  (v4) {};
\node[draw,circle, label=below:{$b$}, fill=black, right = \dist*2 of v4] (v5) {};
\node[draw,circle, fill=white, right = \dist*2 of v5]  (v6) {};
\node[right = \dist of v6] (v7) {$\cdots$};
\node[draw,circle, label=below:$c$, label=above:$w$, right = \dist of v7] (v8) {};

\draw (dummy4) -- (i);
\draw (dummy3) -- (dummy4);
\draw (dummy3) -- (dummy2);
\draw (dummy2) -- (v1);
\draw (v1) -- (v2);
\draw (v2) -- (v3);
\draw (v3) -- (v4);
\draw (v4) -- (v5);
\draw (v5) -- (v6);
\draw (v6) -- (v7);
\draw (v7) -- (v8);
\end{tikzpicture}


\caption{The two cases considered in the proof of \cref{lem:swapPos}. The agent colored in black, $b$, prefer~$y$ to~$w$ while ``dotted'' agents prefer~$w$ over~$y$. The agent~$e$ in the proof is colored gray and also prefers~$w$ over~$y$.}
\label{fig:casedistinction}
\end{figure}

We next define the \myemph{type} of an object.
The type of an object is represented by the index of the edge where the object can possibly be swapped with~$x$.

\begin{definition}\label{def:types}
  Define the \emph{index} of each edge~$\{I+t-1,I+t\}$ to be $t$, $t\ge 1$.  
  For each object~$y$,
  if $y$ and $x$ can possibly be swapped at some edge~(see \cref{lem:swapPos}),
  then let the \myemph{type} of~$y$ be the index of this edge.
  Otherwise, the \myemph{type} of~$y$ is~$0$.
\end{definition}

\begin{figure}[t]
  \centering
\begin{tikzpicture}[scale=0.85, every node/.style={scale=0.9}]
\node[circle,draw, label=below:1, label=above:$a$] at (0,0) (v1) {};
\node[circle,draw, label=below:2, label=above:$b$] at (1,0) (v2) {};
\node[circle,draw, label=below:$I$, label=above:$c$] at (2,0) (v3) {};
\node[circle,draw, label=below:4, label=above:$d$] at (3,0) (v4) {};
\node[circle,draw, label=below:5, label=above:$x$] at (4,0) (v5) {};

\draw (v1) -- (v2);
\draw (v2) -- (v3);
\draw (v3) -- node[above] {1} (v4);
\draw (v4) -- node[above] {2} (v5);
\end{tikzpicture}
\begin{tabular}{llll}
  $1:b \succ$ \fbox{$a$}, & $2: c \succ a \succ$ \fbox{$b$}\,, &  $I:x \succ a \succ$ \fbox{$c$}\,,\\
  $4: a \succ c \succ x \succ$ \fbox{$d$}\,, &
  $5: d \succ$ \fbox{$x$}\,.
\end{tabular}

\caption{An example for types (\cref{def:types}). 
  Object~$b$ can never be swapped with~$x$ before $x$ reaches agent~$I$, and therefore its type is~$0$. The type of object~$d$ is~$2$ and the type of~$a$ and~$c$ is~$1$.}
\label{fig:types}
\end{figure}
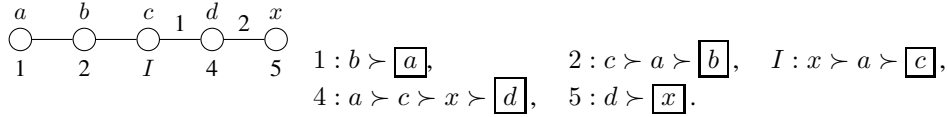

\noindent \Cref{fig:types} shows an example of types.
Note that it is not possible that the edge found by \cref{lem:swapPos} has no index as we can assume that agent~$I$ prefers object~$x$ over all other objects.
Clearly, in every sequence of swaps in which agent~$I$ gets~$x$, object~$x$ is swapped exactly once with an object of each type~$\geq 1$.

\noindent Assume that $z$ the object of type~$1$ that agent~$I$ swaps last to get object~$x$.
Observe that since the underlying graph is a path, each object that initially starts between objects~$x$ and~$z$ must be swapped with exactly one of these two objects in every sequence of swaps that ends with agent~$I$ exchanging~$z$ for~$x$.
In the algorithm we will try all objects of type~$1$ and check whether at least one of them yields a solution.
This only adds a factor of~$n$ to the running time.
Observe that object~$z$ needs to reach~$I$ from the left and hence \cref{lem:swapPos} applies for~$z$ and~$I$.
We will use this fact to show that there are at most two possible candidate objects of each type.
For this, we first define the \myemph{subtype} of an object.
Roughly speaking, the subtype of~$y$ encodes whether~$y$ is left or right of the object of the same type that shall move to the right.

\begin{definition}\label{def:subtypes}
For each object~$y$ of type~$\alpha > 1$, let~$e$ be the edge where~$y$ and~$z$ can possibly be swapped.
If $e$ does not exist, then set the type of all other objects of type~$\alpha$ to~$0$ and set the subtype of~$y$ to~$\ell$.
If $e$ exists, then let~$h$ be the number of edges between the agent~$a$ initially holding~$y$ and the edge~$e$; if~$e$ is incident to~$a$, then let~${h=0}$.
If $\alpha \leq h + 1$, then the \myemph{subtype} of $y$ is ``$r$'' (for right); otherwise the \myemph{subtype} of $y$ is ``$\ell$'' (for left).
\end{definition} 

\Cref{fig:subtypes} shows an example of subtypes.
Notice that if the edge~$e$ exists, then it is unique as stated in \cref{lem:swapPos}. 
\begin{figure}[t]
\centering
\begin{tikzpicture}[scale=0.9, every node/.style={scale=0.9}]
\node[circle,draw, label=above:$z$, label=below:$1$] at (0,0) (v1) {};
\node[circle,draw, label=above:$a$, label=below:$2$] at (1,0) (v2) {};
\node[circle,draw, label=above:$b$, label=below:$3$] at (2,0) (v3) {};
\node[circle,draw, label=above:$c$, label=below:$4$] at (3,0) (v4) {};
\node[circle,draw, label=above:$d$, label=below:$5$] at (4,0) (v5) {};
\node at (5,0) (v6) {$\cdots$};

\draw (v1) -- (v2);
\draw (v2) -- (v3);
\draw (v3) -- (v4);
\draw (v4) -- (v5);
\draw (v5) -- (v6);
\end{tikzpicture}

\begin{tabular}{l@{\;\;}l@{\;\;}l@{\;\;}l}
$1: a \succ$ \fbox{$z$}\,, &  $2: c \succ b \succ z \succ$ \fbox{$a$}\,, & $3: z \succ c \succ$ \fbox{$b$}\,,\\
 $4: d \succ z \succ b \succ$\,, \fbox{$c$}, &
 $5: z \succ$ \fbox{$d$}\,.
\end{tabular}
\caption{An example for subtypes (\cref{def:subtypes}). Preference lists only contain the objects that are depicted in the picture. Assume that objects~$a,b,c,$ and~$d$ are all of type~$2$. If~$c$ has to be swapped with~$z$ at any point, then it is required to move once to the left before being swapped with~$z$ as agent~$3$ prefers~$z$ over~$c$.
  Thus, the value~$h$ as defined in \cref{def:subtypes} is $1$.
  One can verify that $c$ has subtype~$r$.
  Agents~$a$,~$b$, and~$d$ have subtype~$\ell$.}
\label{fig:subtypes}
\end{figure}
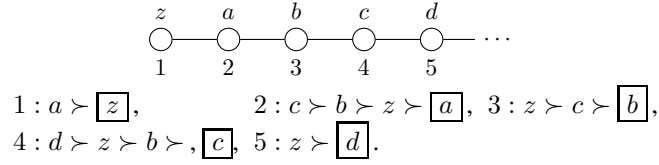

The following auxiliary result helps to identify at most two relevant objects of each type (one of each subtype).
\newcommand{\objectsmovingaround}{
  Consider an object~$y$ that is swapped with object~$z$ before object~$z$ reaches agent~$I$, let $\{i-1,i\}$ be the edge where object~$z$ and $y$ are swapped, and let~$j$ be the agent that initially holds~$y$.
  Then, for each object~$w$ it holds that
  \begin{inparaenum}[(i)]
    \item\label{prop:swapleft} if~$w$ is swapped with~$y$ before~$y$ is swapped with $z$ then~$w$ has type from $\{2,3,\ldots,j-i+1\}$, and
    \item\label{prop:smalltype}  if~$w$ has a type from $\{2,3,\ldots,j-i+1\}$ and is swapped with object~$x$, then it has to be swapped with object~$y$ before~$y$ and~$z$ can be swapped.
    \end{inparaenum}  
}
\begin{lemma}
\label{lem:interval}\objectsmovingaround
\end{lemma}

\begin{proof}
	Let~$y$ be an object that is initially held by some agent~$j$ and that is swapped with~$z$ before~$x$ and~$z$ are swapped.
	Then,~$y$ has to be moved to the left and we can use \cref{lem:swapPos} to compute a unique edge~$\{i-1,i\}$ at which~$y$ and~$z$ are swapped.
	Throughout this proof we will implicitly use the fact that the relative order of all objects that move to the right can never change as otherwise an agent would regain an object she already gave away before and hence she would have made an irrational trade.
 
	Statement~\eqref{prop:swapleft}:
	Let~$w$ be an object that is swapped with~$y$ before~$y$ is swapped with $z$.
	Then~$w$ has to be moved to the right.
	Suppose towards a contradiction that~$w$ has type at least~$j-i+2$; note that no object can have type~$1$ as object~$z$ will be the type-$1$ object.
	Then, either there are at least~$j-i$ other objects (of types~$2,3,\ldots j-i+1$) that initially start between~$y$ and~$z$ and that are moved to the right or there is a type~$\alpha \in [2,j-i+1]$ such that no object of type~$\alpha$ is initially between~$y$ and~$z$ and that is moved to the right.
	In the former case, there are at least~$j-i+1$ objects that are initially between~$y$ and~$z$ and that move to the right and hence~$y$ is moved to agent~$i-1$ before it can be swapped with~$z$, a contradiction.
	In the latter case, note that after objects~$w$ and~$x$ are swapped, there is some~$\alpha \in [2,j-i+1]$ such that no object of type~$\alpha$ is between~$x$ and~$z$ and that is moved to the right and hence~$x$ can not be swapped over the edge~$\{I+\alpha-1,I+\alpha\}$, again a contradiction.
	Thus,~$w$ has a type at most~$j-i+1$.

	Statement~\eqref{prop:smalltype}:
	Let~$w$ have a type~$\alpha \in [2,j-i+1]$ and let~$w$ be moved to the right.
	Suppose towards a contradiction that~$w$ and~$y$ are not swapped.
	Since~$y$ is moved to the left and~$w$ is moved to the right,~$w$ has to start to the right of~$y$.
	Since the relative order of objects moving to the right cannot change and since object~$x$ is swapped with all objects moving to the right in decreasing order of their types, it follows that all objects that move to the right are initially ordered by their type.
	Hence, there are only objects of type~$\beta \in [2,\alpha-1]$ that initially start between~$w$ and~$z$ and that move to the right.
	Since~$y$ is initially left of~$w$, it holds that there are at most~$j-i-1$ objects that are initially between~$y$ and~$z$ and that are moved to the right.
	Thus,~$y$ can only reach agent~$i+1$ before it has to be swapped with~$z$.
	Since the edge computed by \cref{lem:swapPos} is unique,~$y$ and~$z$ cannot be swapped over~$\{i,i+1\}$, a contradiction.
	Thus,~$w$ and~$y$ have to be swapped before~$y$ and~$z$ can be swapped.
      \end{proof}

\noindent Subtypes help to exclude all but two objects of each type.

\begin{lemma}
\label{lem:type}
Given objects~$x$ and~$z$, there is an~$O(n^2)$-time preprocessing that excludes all but at most two objects of each type~$\alpha \geq 2$ as potential candidates for being swapped with~$x$.
\end{lemma} 

\begin{proof}
Consider a type~$\alpha \geq 2$ and all objects of type~$\alpha$.
Compute the subtype of each of these objects.
Exactly one of them is swapped to the right and all others have to be swapped with~$z$ at some point.
From \cref{lem:swapPos}, we know that all the objects that are swapped to the left have a specific edge where they can possibly be swapped with~$z$.
If such an edge does not exist for some object, then we know that this object has to move to the right and we can change the type of all other objects of type~$\alpha$ to~$0$.
Note that there is no solution if such an edge does not exist for multiple objects of the same type.

If such an edge exists for each object of type~$\alpha$, then we count the number~$h$ of swaps to the left that are needed for each object~$y$ of type~$\alpha$ to reach the edge at which it can be swapped with~$z$.
Note that by \cref{lem:interval} each of these swaps happens with an object 
of type~$\beta \in [2,h+1]$.
If these types include $\alpha$, that is, $\alpha \leq h+1$, then~$y$ has subtype~$r$ and otherwise it has subtype~$\ell$.
If the subtype of~$y$ is~$r$ and if~$y$ is moved to the left, then by \cref{lem:interval}~$y$ has to be right of the object of type~$\alpha$ that moves to the right. 
Again by \cref{lem:interval}, if~$y$ is moved to the left and has subtype~$\ell$, then~$y$ has to be left of the object of type~$\alpha$ that is moved to the right. 

Consider the case where the objects of a given type are not ordered by their subtype (from left to right:~$\ell,\ell,\ldots,\ell,r,r,\ldots,r$).
Then, for each object~$w$ of type~$\alpha$, there exists an object of type~$\alpha$ and subtype~$r$ to the left or an object of type~$\alpha$ and subtype~$\ell$ to the right.
In \cref{fig:subtypes}, for~$w = d$ there is~$c$ with subtype~$r$ left of~$w$ and for~$w \in \{a,b,c\}$ there is object~$d$ of subtype~$\ell$ right of~$w$.
Thus, if we try to send~$w$ to the right, then the number of swaps to the left of some other object of type~$\alpha$ (objects~$c$ or~$d$ in \cref{fig:subtypes}) does not match the number of swaps needed to reach the edge where the objects can be swapped with~$z$.
Hence, there is no solution.

Now consider the case where the objects are ordered by subtype as indicated above.
By the same argument as above there are only two possible objects of type~$\alpha$ that can possibly travel to the right: The last object of subtype~$\ell$ and the first object of subtype~$r$.
We can therefore set the type of all other objects of type~$\alpha$ to~$0$.

Let~$n_\alpha$ be the number of objects of type~$\alpha$.
Since the subtype for each object of type~$\alpha$ can be computed in~$O(n)$ time, 
we obtain that the described preprocessing takes~$O(n_\alpha \cdot n)$ time for type~$\alpha$.
After having computed the subtype of each object of type~$\alpha$, we iterate over all these objects and find the two specified objects or
determine that the objects are not ordered by subtype in~$O(n)$ time.
Hence the overall running time is in~$O(\sum_{\alpha>1} (n_\alpha \cdot n)) \subseteq O(n^2)$.
The inclusion holds since each object (except for~$x$) has exactly one type.
\end{proof}

We are now in a position to present the heart of our proof.
We will show how to choose an object of each type~$\alpha \geq 1$ such that moving those objects to the right and all other objects to the left leads to a swap sequence such that agent~$I$ gets object~$x$ in the end.
Once we have chosen the correct objects, we can compute the final position of each object in linear time and then use the fact then any swap sequence that only sends objects ``in the correct'' direction is a valid sequence since the relative order of all objects that travel to the left (respectively to the right) can never change.
Such a selection leads to a solution if and only if for each pair of objects such that the right one moves to the left and the left one moves two the right, the two endpoints of the edge where they are to be swapped can agree on this swap.
We mention that these insights were also used in Algorithm~1 by \citet{GouLesWil2017}.

We will next focus on objects of type~$0$.
Using \cref{lem:swapPos}, we can compute for each object~$y$ of type~$0$ the edge where~$y$ and~$z$ can be swapped.
If such an edge does not exist, then there is no solution.
Hence, we can again compute the number~$h$ of objects between~$y$ and~$z$ that need to move to the right.
If any object~$w$ which is to the right of~$y$ has a type~$\beta \leq h+1$ or any object~$w'$ to the left of~$y$ has a type~$\beta' > h + 1$, then by \cref{lem:interval} these object~$w$ and~$w'$ cannot be moved to the right and hence we can set its type to~$0$.
Suppose, towards a contradiction, that there is an object~$y$ of type~$0$ between two objects~$v,w$, both of type~$\alpha$.
Then, we can either set the type of~$w$ (if~$\alpha \leq h+1$) or of~$v$ (if~$\alpha > h+1$) to~$0$.
Hence, there is no object of type~$0$ between two objects of the same type and we are guaranteed that~$z$ can be swapped with all objects of type~$0$ regardless of which objects of each type are moved to the right.

Hence, it remains to study swaps (i) of~$z$ with objects of type at least~$2$ that move to the left, (ii) of objects that move to the right and objects of type at least~$2$ that move to the left, and (iii) of objects of type~$0$ and objects moving to the right.
Before doing so, we need to define the last ingredient for our proof: \myemph{blocks}.

\begin{definition}\label{def:blocks}
  A \myemph{block} is a minimal subset~$B\subseteq X$ of objects that contains all objects of all types in some interval~$[\alpha,\beta]$ with~$2 \le \alpha \le \beta$ such that all objects in~$B$ are initially hold by agents on a subpath of the input path (see \cref{fig:blocks}).
\end{definition}

\begin{figure}
\centering
\begin{tikzpicture}[scale=0.75, every node/.style={scale=1}]
\node at (-1,0) (v0) {$\cdots$};
\node[circle,draw, label=above:$z$] at (0,0) (v1) {};
\node[circle,draw, label=above:$2_\ell$] at (1,0) (v2) {};
\node[circle,draw, label=above:$3_\ell$] at (2,0) (v3) {};
\node[circle,draw, label=above:$2_r$] at (3,0) (v4) {};
\node[circle,draw, label=above:$4_\ell$] at (4,0) (v5) {};
\node[circle,draw, label=above:$3_r$] at (5,0) (v6) {};
\node[circle,draw, label=above:$4_r$] at (6,0) (v7) {};
\node[circle,draw, label=above:$0$] at (7,0) (v8) {};
\node[circle,draw, label=above:$5_\ell$] at (8,0) (v9) {};
\node[circle,draw, label=above:$5_r$] at (9,0) (v10) {};
\node at (10,0) (v11) {$\cdots$};

\draw (v0) -- (v1);
\draw (v1) -- (v2);
\draw (v2) -- (v3);
\draw (v3) -- (v4);
\draw (v4) -- (v5);
\draw (v5) -- (v6);
\draw (v6) -- (v7);
\draw (v7) -- (v8);
\draw (v8) -- (v9);
\draw (v9) -- (v10);
\draw (v10) -- (v11);

\draw (0.5,1) -- (6.5,1) -- (6.5,-0.5) -- (0.5,-0.5) -- cycle;
\draw (7.5,1) -- (9.5,1) -- (9.5,-0.5) -- (7.5,-0.5) -- cycle;
\end{tikzpicture}
\caption{An example of blocks~(\cref{def:blocks}). The left block contains all objects of type~$2,3$ and~$4$ and the right block contains the two objects of type~$5$.}
\label{fig:blocks}
\end{figure}
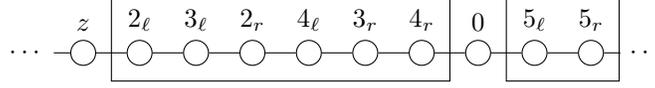


We first prove that blocks are well defined and that each object of type~$\eta \geq 2$ is contained in exactly one block.

\begin{lemma}
\label{lem:block}
Each object of type~$\eta \geq 2$ is contained in exactly one block, if there is no object of a higher type to its left, then this object has the highest type in its block, and all blocks can be computed in linear time.
\end{lemma}

\begin{proof}
For any type~$\delta$, let~$\delta_\ell$ be the object of type~$\delta$ and subtype~$\ell$ and analogously let~$\delta_r$ be the object of type~$\delta$ and subtype~$r$.
Recall that~$\delta_\ell$ is left of~$\delta_r$.
If there is only one object~$w$ of type~$\delta$, then we say that~$\delta_\ell = \delta_r = w$.
Observe that the objects in a block are initially hold by agents of a subpath of the input path and that the subpaths of different blocks do not intersect.
The ``first'' or ``leftmost'' block with respect to this subpath has to start with~$2_\ell$ as otherwise the leftmost object of the block (again with respect to its subpath) could never be chosen and hence its type could be set to~$0$.
Analogously, the objects of a block with interval~$[\alpha,\beta]$ have to start with~$\alpha_\ell$ and end with~$\beta_r$, as otherwise the leftmost (respectively rightmost) object could never be chosen since there is no object of type~$\alpha$ (of type~$\beta$) to its left (respectively to its right) and its type is larger than~$\alpha$ (less than~$\beta$).

We will start with the block that contains type~$2$ in its interval and show that there is a unique type~$\gamma$ such that the block has interval~$[2,\gamma]$.
Applying the same argumentation iteratively with~$\gamma + 1$, we show that each object of type~$\eta > 2$ is contained in exactly one block:
Consider the smallest type~$\alpha$ that is not yet shown to be in a block and consider the objects~$\alpha_\ell$ and~$\alpha_r$.
If these two objects are the same or are initially hold by adjacent agents, then these object(s) are a block with interval~$[\alpha,\alpha]$ and since blocks are minimal sets of objects, they are not part of any other block.
If initially there is an object between~$\alpha_\ell$ and~$\alpha_r$, then it cannot be objects of type~$\eta < \alpha$.
This holds, since it cannot be an object of type~$0$ as shown above, it cannot be~$z$ as~$z$ is the leftmost object that we consider and by assumption it cannot be an object of type~$\eta \in [2,\alpha-1]$.
The right neighbor of~$\alpha_\ell$ therefore has to be~$(\alpha+1)_\ell$ since if it was of type~$\eta > \alpha+1$, then it could never be chosen as there is no object of type~$\alpha +1$ to its left and hence its type could be set to~$0$.
Analogously, note that~$(\alpha+1)_r$ has to initially be right of object~$\alpha_r$ as otherwise object~$\alpha_r$ could never be chosen.
Hence we can continue with the objects~$(\alpha+1)_\ell$ and~$(\alpha+1)_r$.
Again, there can be no objects of type~$\eta < \alpha$ between them.
If there are only objects of type~$\alpha$ between them, then the block has interval~$[\alpha,\alpha+1]$ and otherwise we can continue with objects~$(\alpha+2)_\ell$ and~$(\alpha+2)_r$ and so on.
Note that this chain stops exactly at the first type~$\eta$ where all objects of a higher type than~$\eta$ are initially right of~$\eta_r$ and this happens at latest at~$\delta_r$, where~$\delta$ describes the largest type.
Since we need only a constant amount of computation for each object, all blocks can be computed in linear time.
\end{proof}

We next prove that for each block there are only two possibilities to choose objects of each type in the block that can lead to a solution.
We start with an intermediate lemma.

\newcommand{\ftwol}{%
If in a block with interval~$[\alpha,\beta]$, we decide for some type~$\gamma \in [\alpha,\beta]$ to send object~$\gamma_r$ to the  right, then we need to send all objects~$\delta_r$ of type~$\delta \in [\gamma,\beta]$ to the right. %
}
\begin{lemma}[\appsymb]
\label{lem:f2l}
\ftwol
\end{lemma}

\begin{proof}
Let~$B$ be a block with interval~$[\alpha,\beta]$ and let~$\gamma\in[\alpha,\beta]$ be some type and assume that~$\gamma_r$ is to be sent to the right.
By \cref{lem:block}, we know that unless~$\gamma = \beta$, it holds that~$(\gamma+1)_\ell$ is to the left of~$\gamma_r$ and can therefore not be sent to the right.
Thus, we also have to move~$(\gamma+1)_r$ to the right.
This argument applies iteratively for all types in~$[\gamma,\beta]$.
\end{proof}

Based on \Cref{lem:f2l}, we prove the following.

\newcommand{\twopos}{%
There are at most two selections of objects in a block with interval~$[\alpha,\beta]$ that can lead to~$I$ getting~$x$. These selections can be computed in~$O(n\!\cdot\! (\beta-\alpha+1))$~time.%
}
\begin{lemma}
  \label{lem:2pos}
  \twopos
\end{lemma}

\begin{proof}
Let~$B$ be a block with interval~$[\alpha,\beta]$.
There are only two possibilities for type~$\alpha$: Either~$\alpha_\ell$ is moved to the right or~$\alpha_r$ is moved to the right.
If~$\alpha_r$ is moved to the right, then \cref{lem:f2l} states that we need to move~$\delta_r$ to the right for all~$\delta \in [\alpha,\beta]$.
If we want to move~$\alpha_\ell$ to the right, then we know the final destination of~$\alpha_\ell$ and that~$\alpha_\ell$ and~$\alpha_r$ have to swap at some point.
We can therefore use \cref{lem:swapPos} to compute the edge where~$\alpha_\ell$ and~$\alpha_r$ are swapped.
From this we can compute the number~$h$ of objects between~$\alpha_\ell$ and~$\alpha_r$ that have to be moved to the right and hence we know that they have to be of types~$\alpha+1,\alpha+2,\ldots,\alpha+h$.
Note that no object of subtype~$r$ between~$\alpha_\ell$ and~$\alpha_r$ can be moved to the right.

Thus, there are two possibilities: The number~$h$ equals the number of objects of subtype~$r$ between~$\alpha_\ell$ and~$\alpha_r$ or not.
In the former case we know that~$h$ objects of subtype~$\ell$ have to be moved to the right and then an object of subtype~$r$ (the respective object of subtype~$\ell$ is to the left of~$\alpha_r$).
By \cref{lem:f2l}, all remaining objects of subtype~$r$ have to be moved to the right.
In the latter case we can consider the highest type~$\eta$ that occurs between~$\alpha_\ell$ and~$\alpha_r$.
If~$\eta = \beta$, then we have to only move objects of subtype~$\ell$ to the right and otherwise we can next look at the objects between~$\eta_\ell$ and~$\eta_r$ and apply the very same argument as before.
Thus, there is only the choice at the very beginning whether to move~$\alpha_\ell$ or~$\alpha_r$ to the right.
One possibility is to only move objects of subtype~$r$ to the right and the other can be computed in~$O(n \cdot  (\beta-\alpha+1))$ time as there are~$O(\beta-\alpha+1)$ types in~$B$ and for each type we can compute the possible swap position in~$O(n)$ time.
All other computations can be done in constant time per object.
\end{proof}

For both possible selections (see \Cref{lem:2pos}), we can determine in~$O((\beta-\alpha+1)^2 \cdot n)$ time, whether this selection is \myemph{consistent}, i.e., any pair of objects in this block that needs to be swapped at some point can in fact be swapped, by computing for each pair of objects where they should be swapped in~$O(n)$ time using \cref{lem:swapPos} (observe that we know the final destination of the object moving to the right) and checking whether the two endpoints of this edge can agree on this.
We further require that this selection is not in conflict with objects of type~$0$ in the sense that all of the objects that we move to the right can be swapped with all objects of type~$0$.
Observe that objects of type~$0$ are initially not located between objects of the same type and therefore we know exactly where the objects of the selection and the objects of type~$0$ are swapped independent of the selection for other blocks.

By the definition of types, we know that~$x$ can always be swapped with the objects we moved to the right.
Consider a possible selection of objects from some block~$B$ to move to the right.
All other objects are moved to the left and hence have to be swapped with~$z$ at some edge.
Since we know the number of objects to the left of~$B$ that are moved to the right (we do not know which objects these are but we know their number and types by the definition of blocks), we can compute for each of them the edge where they need to swap with~$z$.
If the swap of the considered object and~$z$ is not rational for the two endpoints of this edge, then the selection can never lead to a situation where~$I$ swaps~$z$ for~$x$ and we can therefore ignore this selection.

Thus, it only remains to find a selection for each block such that the objects that are moved to the right and the objects that are moved to the left can be swapped (if the former one is initially to the left of the latter one).
We say that these selections are \myemph{compatible}.

\begin{definition}\label{def:compatible}
Let~$B$ and $C$ be two blocks with intervals~$[\alpha,\beta]$ and~$[\gamma,\delta]$, respectively, and let~$\beta < \gamma$.
Let~$s_B$ (resp.~$s_C$) be a selection of objects from~$B$ (respectively $C$) to move to the right.
We say that~$s_B$ and~$s_C$ are \myemph{compatible} if for all~$b \in s_B$ and all~$c \in C \setminus s_C$, the swap of~$b$ and~$c$ is rational for the two agents at the (unique) position where~$b$ and~$c$ can be swapped.
Otherwise, we say that~$s_B$ and~$s_C$ are \myemph{in conflict}.
\end{definition}

\noindent Observe that we can compute a unique pair of agents that can possibly swap~$b$ and~$c$ since we know how many objects between~$b$ and~$c$ are moved to the right.
This number is the sum of objects in~$s_B$ right of~$b$, objects in~$s_C$ left of~$c$, and $\gamma - \beta - 1$.
Hence, computing whether these two selections are compatible takes~$O(|B| \cdot |C| \cdot n)$ time as the agents can be computed in constant time per pair of objects and checking whether these agents can agree on swapping takes~$O(n)$ time.
It remains to find a selection for each block such that all of these selections are pairwise compatible.
We solve this problem using a reduction to 2-SAT which is known to be linear-time solvable~\cite{APT79}.

\begin{theorem}
  \label{thm:path}
\ROs{} on paths can be solved in~$O(n^4)$ time. 
\end{theorem}

\begin{proof}
  For each object~$z$ of type~$n-I$ (there are $O(n)$ many), we do the following.
  First, compute the type, subtype and block of each object and use \cref{lem:2pos} to compute two possible selections for each block in overall~$O(n^2)$ time.
  Second, compute in~$O(n^3)$ time pairs of selections compatible to each other.
  Third, check in~$O(n^3)$ time whether these selections are consistent.

If some selection~$s$ for a block~$B$ is not consistent or if there is some other block~$C$ such that~$s$ is not compatible with either selection for~$C$, then we know we have to take the other possible selection~$s'$ for~$B$ and can therefore ignore all selections that are in conflict with~$s'$.
If this rules out some selection, then we can repeat the process.
After at most~$n$ rounds of which each only takes~$O(n)$ time, we arrive at a situation where there are exactly two consistent selections for each block and the task is to find a set of pairwise compatible selections that include a selection for each block.
We finally reduce this problem to a 2-SAT formula.

We start with a variable~$v_B$ for each block~$B$ which is set to true if we move all objects of subtype $r$ to the right and false otherwise.
For each pair~$s,s'$ of selections that are in conflict with one another, let~$B$ be the block of~$s$ and~$C$ be the block of~$s'$.
Without loss of generality, let~$s$ and~$s'$ be the selections that are represented by~$v_B,v_C$ being set to true (otherwise swap~$\lnot v_B$ with~$v_B$ or~$\lnot v_C$ with~$v_C$ in the following clause).
Since we cannot select~$s$ and~$s'$ at the same time, we add a clause~$(\lnot v_B \lor \lnot v_C)$ to our 2-SAT formula.

\noindent Observe that if there is a set of pairwise non-conflicting selections, then the 2-SAT formula is satisfied by the corresponding assignment of the variables and if the formula is satisfied, then this assignment corresponds to a solution to the original \ROs-instance.
Since 2-SAT can be solved in linear time~\cite{APT79} and the constructed formula has~$O(n^2)$ clauses of constant size, our statement follows. 
\end{proof}

\noindent We conclude by conjecturing that the case when the underlying graph is a cycle can also be solved in polynomial time.
The idea is similar to the case of a path.
The main difference is that it may happen that some objects may be swapped with~$x$ twice.
Since we ``guess'' the object~$z$ that is last swapped with~$x$ (and we can also ``guess'' the moving direction of~$x$ in a similar fashion), we can compute the first edge where~$x$ and~$z$ are swapped.
This determines the number~$k$ of objects which initially start between~$x$ and~$z$ and are also required to swap twice with~$x$ in any solution.
We then apply the same type-based arguments as in the proof for paths, but we incorporate the additional information that the objects of types~$1,\ldots,k+1$ that are swapped with~$x$ are also the objects with the last (largest)~$k+1$ types that are swapped with~$x$.

\section{Preferences of  Length at Most Four}
\label{sec:complete-graphs}

In this section we investigate the case where we do not impose any restriction on the underlying social network, i.e., it is a complete graph. 
We find that \ROs{} remains NP-complete in this case. 
This implies that the computational hardness of the problem does not stem from restricting the possible swaps between agents by an underlying social network.
Moreover, the hardness holds even if each agent has at most four objects in her preference list.
To show NP-hardness, we reduce from a restricted NP-complete variant of the \threesat{} problem~\cite{Tovey84}.
In this variant, each clause has either $2$ or $3$ literals, and each variable
  appears once as a negative literal and either once or twice as a positive literal.
  We note that in the original NP-hardness reduction by \citet{GouLesWil2017}, the lengths of the preference lists are unbounded.

  \newcommand{\thmnpclengthfour}{
  \ROs{} is NP-complete on complete graphs, even if each preference list has length at most four.    
  }

\begin{theorem}\label[theorem]{thm:NP-c-length-4}
  \thmnpclengthfour
\end{theorem}

\begin{proof}
  We only focus on the hardness reduction as containment in NP is shown in~\cref{prop:RO-in-NP}.
  We reduce from the restricted \threesat{} variant mentioned in the beginning of the section. 
  The general idea of the reduction is to introduce for each literal of a clause a pair of \emph{private} clause agents to pass through the target object if the corresponding literal is set to $\ttrue$
  and to introduce for each variable some variable agents to make sure that no two pairs of private clause agents for which the corresponding literals are complement to each other will pass through the target object in the same sequence of swaps. 
  In this way, we can identify a satisfying truth assignment if and only if there is a sequence of swaps that makes the target object reach our agent.
  Let $\phi=(\mathcal{V}, \mathcal{C})$ be an instance of the restricted \threesat{} problem
  with variables~$\mathcal{V}=\{v_1,\ldots, v_n\}$ and clauses~$\mathcal{C}=\{C_1,\ldots, C_m\}$.
  For each variable~$v_i\in \mathcal{V}$,
  let $\occ(i)$ be the number of occurrences of variable~$v_i$ (note that $\occ(i)\in\{2,3\}$),
  let $\nu(i)$ denote the index of the clause that contains the negative literal~$\negv_i$,
  and let $\pi_1(i)$ and $\pi_2(i)$ be the indices of the clauses with $\pi_1(i) < \pi_2(i)$ that contain
  the positive literal~$v_i$; if $v_i$ only appears twice in $\phi$, then we simply neglect~$\pi_2(i)$.
  Now, we construct an instance of \ROs{} as follows.

  \myparagraph{Agents and Initial Assignment~$\sigma_0$.}
  For each variable~$v_i\in \mathcal{V}$, introduce $\occ(i)\!-\!1$ \myemph{variable agents}, denoted as $X_i^{z}$ with initial objects~$x_i^{\occ(i)-z}$, $z\in \{1,\ldots, \occ(i)\!-\!1\}$.
  For each clause~$C_j \in \mathcal{C}$, introduce $2|C_j|+1$ \myemph{clause agents}, denoted as $A_j$, $B_j^{z}$, and $D_j^{z}$, $z\in \{1,\ldots, |C_j|\}$, where $|C_j|$ denotes the number of literals contained in~$C_j$.
  The initial objects of $B^z_j$, and $D_j^{z}$ are $b_j^z$, and $d_j^z$, respectively.
  The initial object of $A_1$ is our target object~$x$ and the initial object of $A_{j}$, $j\ge 2$, is $a_{j-1}$.
  Finally, our target agent~$I$ initially holds object~$a_m$.

  \myparagraph{Preference Lists.}
  For each clause~$C_j\in \mathcal{C}$, we use an arbitrary but fixed order of the literals in $C_j$ to define a bijective function~$\ff_j \colon C_j \to  \{1,\ldots, |C_j|\}$, which assigns to each literal contained in $C_j$ a distinct number from $\{1,\ldots, |C_j|\}$.

  \begin{compactenum}[(1)]
    \item For each variable~$v_i\in \mathcal{V}$, let $j=\nu(i)$, $j'=\pi_1(i)$ (and $j''=\pi_2(i)$ if~$\occ(i)=3$)
    and do the following:

    \begin{compactenum}[(i)]
      \item If $\occ(i)=2$, then $X_i^1$ has preference list
      \begin{align*}
        d_{j'}^{\ff_{j'}(v_i)} \succ d_{j}^{\ff_{j}(\negv_i)} \succ \fbox{$x_i^1$}\,.
        \end{align*}

      \item If $\occ(i) = 3$, then the preference lists of $X_i^1$ and $X_i^2$ are
      \begin{align*}
       X^1_i \colon &  d_{j'}^{\ff_{j'}(v_i)} \succ x_i^1 \succ d_{j}^{\ff_{j}(\negv_i)}  \succ
        \fbox{$x_i^2$}\,,\\
        X^2_i \colon & d_{j''}^{\ff_{j''}(v_i)} \succ x_i^2 \succ~\fbox{$x_i^1$}\,.
      \end{align*}
      \end{compactenum}

    \item We define an auxiliary function to identify an object for each clause~$C_j\in \mathcal{C}$ and each literal~$\ell\in C_j$ contained in~$C_j$:    
    \begin{align*}  \tau(C_j,\ell) \coloneqq
      \begin{cases}
        x_i^1, &\text{if } \occ(i)=2 \text{ and } \ell=\negv_i \text{ for some variable } v_i, \\
        x_i^2, &\text{if } \occ(i)=3 \text{ and } \ell=\negv_i \text{ for some variable } v_i, \\
        x_i^1, &\text{if } \ell = v_i \text{ and } j=\pi_1(i) \text{ for some variable } v_i,\\
        x_i^2, &\text{if } \ell = v_i \text{ and } j=\pi_2(i)\text{ for some variable } v_i.
      \end{cases}
      \end{align*}
      
      The preference lists of the clause agents corresponding to $C_1$ are:
      \begin{align*}
        A_1 \colon   b_1^1 \succ \cdots \succ b_1^{|C_1|} \succ \fbox{$x$}\,.
      \end{align*}
    For each literal~$\ell \in C_1$, the preference lists of $B_1^{\ff_1(\ell)}$ and $D_1^{\ff_1(\ell)}$ are
    \begin{align*}
      B_1^{\ff_1(\ell)} \colon & \tau(C_1, \ell) \succ x  \succ \fbox{$b_1^{\ff_1(\ell)}$},\text{ and }\\ 
      D_1^{\ff_1(\ell)} \colon &  a_1 \succ x \succ \tau(C_1, \ell) \succ \fbox{$ d_1^{\ff_1(\ell)}$}\,.
    \end{align*}
    For each index~$j\in \{2,\ldots, m\}$, the preference lists of the clause agents corresponding to $C_j$ are
    \begin{align*}
      A_j \colon & b_{j}^1 \succ \cdots \succ b_j^{|C_j|} \succ \fbox{$a_{j-1}$},\\ 
      B_j^{\ff_j(\ell)} \colon& \tau(C_j, \ell) \succ x \succ a_{j-1} \succ \fbox{$b_j^{\ff_j(\ell)}$}, \text{ and }\\
      \text{for all~}\ell \in C_j,\text{ let } D_j^{\ff_j(\ell)} \colon&  a_j \succ x \succ \tau(C_j, \ell) \succ \fbox{$d_j^{\ff_j(\ell)}$}\,.
      \end{align*}

    \item Let the preference list of our target agent~$I$ be $x \succ$ \fbox{$a_m$}.
\end{compactenum}

To finish the construction, we let the underlying graph be complete.
    One can verify that the constructed preference lists have length at most four.

  \myparagraph{The underlying graph~$G=(V,\binom{V}{2})$.} All agents are pairwise connected by an edge in the underlying graph.
  By the definition of rational trades, indeed, we can delete all irrelevant edges, say~$\{u,v\}$,
  if $u$ and $v$ will never agree to trade, i.e., there are no two objects, say $i,j$, which exist in the preference lists of both~$u$ and $v$ such that $u$ prefers~$i$ to~$j$ while $v$ prefers~$j$ to~$i$.
  By carefully examining the preference lists of the agents,
  we observe that only the following edges~$E$ are relevant for $V$.
  \begin{compactenum}[(1)]
    \item\label{edge:A-B-D}  For each clause~$C_j$ the corresponding vertices form a \myemph{generalized star} with $A_j$ being the center and each leaf~$D_j^z$ having distance two to the center. 
    Formally, for each clause~$C_j\in \mathcal{C}$ and for each two clause agents~$B_j^z$, $D_j^z$ with $1\le z \le |C_j|$ let $\{A_j, B_j^z\}, \{B_j^z, D_j^z\} \in E$.

    \item\label{edge:B-B} For each $j\in \{1,\ldots, m-1\}$, the two vertex sets~$\{D_{j}^{z} \mid 1\le z \le |C_{j}|\}$ and $\{B_{j+1}^{z'} \mid 1\le z' \le |C_{j+1}|\}$ form a complete bipartite graph.
    Formally, for each $j\in \{1,\ldots, m-1\}$,
    and for each two clause agents~$D_{j}^{z}$ and $B_{j+1}^{z'}$ with $1\le z \le |C_{j}|$ and $1\le z' \le |C_{j+1}|$, let $\{D_{j}^z, B_{j+1}^{z'}\}\in E$. 
    \item To connect the clause agents and variable agents, for each variable~$v_i\in \mathcal{V}$, we do the following.
    \begin{compactenum}[(a)]
      \item  If $\occ(i)=2$, then add $\{X_i^1, B_{j'}^{f_{j'}(v_i)}\}$ and $ \{X_i^1, B_{j}^{f_j(\negv_i)}\}$ to~$E$, where $j=\nu(i)$ and $j'=\pi_1(i)$.
      \item  If $\occ(i)=3$, then add $\{X_i^1, B_{j'}^{f_{j'}(v_i)}\}$, $\{X_i^1, B_{j}^{f_j(\negv_i)}\}$, $ \{X_i^2, B_{j''}^{f_{j''}(v_i)}\}$, and $\{X_i^1,X_i^2\}$ to~$E$, where $j=\nu(i)$, $j'=\pi_1(i)$, and $j''=\pi_2(i)$.
    \end{compactenum}
      \item Our target agent~$I$ is adjacent to all clause agents~$D_m^z$, $z\in \{1, \ldots, |C_m|\}$.
    \end{compactenum}

  \allowdisplaybreaks[1]

    \begin{example}\label{ex:NP-c-l-4}
    For an illustration of the construction, let us consider the following restricted \threesat{} instance:
    \begin{alignat*}{4}
      \mathcal{V}=\{v_1, v_2, v_3, v_4\}, &  \mathcal{C} = \{C_1=(v_2\vee v_3), C_2 =(v_1\vee \negv_2 \vee \negv_3),    C_3 =(\negv_1\vee v_2 \vee v_4),   C_4 =(v_3\vee \negv_4) \}.
    \end{alignat*}
    Our instance for \ROs{} contains the following agents.
    \begin{align*}
      V=& \{A_1,B_1^1, B_1^2, D_1^1, D_1^2\} \cup \{A_2,B_2^1,B_2^2,B_2^3, D_2^1,D_2^2, D_2^3\} \cup \\
        &\{A_3,B_3^1,B_3^2,B_3^3,D_3^1,D_3^2,D_3^3\}\cup \{A_4,B_4^1, B_4^2,D_4^1,D_4^2\} \cup\\
          &\{X_1^1,X_2^1,X_2^2,X_3^1,X_3^2,X_4^1\} \cup \{I\}.
    \end{align*}
    The preference lists of these agents are

    \begin{table*}
    \begin{tabular}{r@{}lr@{}lr@{}lr@{}lr@{}}
      $A_1\colon$ & $b_1^1 \succ b_1^2 \succ$ \fbox{$x$},
      &~~$A_2\colon$ &$b_2^1 \succ b_2^2  \succ b_2^3 \succ$ \fbox{$a_1$},
      &~~$A_3\colon$ & $b_3^1 \succ b_3^2 \succ b_3^3 \succ$ \fbox{$a_2$},
      &~~$A_4\colon$ & $b_4^1 \succ b_4^2 \succ$ \fbox{$a_3$},\\
      $B_1^1\colon$ & $x_2^1 \succ x \succ$ \fbox{$b_1^1$}, & $B_2^1\colon$& $x_1^1 \succ x \succ a_1 \succ$ \fbox{$b_2^1$}, & $B_3^1\colon$ & $x_1^1 \succ x \succ a_2 \succ$ \fbox{$b_3^1$}, & $B_4^1\colon$& $x_3^2 \succ x \succ a_3 \succ$ \fbox{$b_4^1$},\\
      $B_1^2\colon$& $x_3^1 \succ x \succ$ \fbox{$b_1^2$}, & $B_2^2\colon$& $x_2^2 \succ x \succ a_1 \succ$ \fbox{$b_2^2$}, & $B_3^2\colon$& $x_2^2 \succ x \succ a_2 \succ$ \fbox{$b_3^2$},
      & $B_4^2\colon$& $x_4^1 \succ x \succ a_3 \succ$ \fbox{$b_4^2$},\\      
      & & $B_2^3\colon$& $x_3^2 \succ x \succ a_1 \succ$ \fbox{$b_2^3$},& $B_3^3\colon$& $x_4^1 \succ x \succ a_2 \succ$ \fbox{$b_3^3$},\\
      $D_1^1\colon$& $a_1 \succ x\succ x_2^1 \succ$ \fbox{$d_1^1$}, &
      $D_2^1\colon$& $a_2\succ x \succ x_1^1 \succ$ \fbox{$d_2^1$}, &
      $D_3^1\colon$& $a_3\succ x \succ x_1^1 \succ$ \fbox{$d_3^1$}, &
      $D_4^1\colon$& $a_4\succ x \succ x_3^2 \succ$ \fbox{$d_4^1$}, \\
      $D_1^2\colon$& $a_1\succ x \succ x_3^1  \succ$ \fbox{$d_1^2$}, &
      $D_2^2\colon$& $a_2\succ x \succ x_2^2  \succ$ \fbox{$d_2^2$}, &
      $D_3^2\colon$& $a_3\succ x \succ x_2^2  \succ$ \fbox{$d_3^2$}, &
      $D_4^2\colon$& $a_4\succ x \succ x_4^1  \succ$ \fbox{$d_4^2$},\\
      &&  $D_2^3\colon$ & $a_2\succ x \succ x_3^2 \succ$ \fbox{$d_2^3$},&  $D_3^3\colon$ & $a_3\succ x \succ x_4^1 \succ$ \fbox{$d_3^3$},\\
      $I \colon$ & $x \succ$ \fbox{$a_4$}, \\[2ex]
      $X_1^1\colon$ & $d_2^1 \succ d_3^1 \succ$ \fbox{$x_1^1$}, & $X_2^1 \colon$ & $d_1^1\succ x_2^1 \succ d_2^2 \succ$ \fbox{$x_2^2$}, & $X_3^1 \colon$ & $d_1^2 \succ x_3^1 \succ d_2^3 \succ$ \fbox{$x_3^2$}, & $X_4^1\colon$ & $d_3^3 \succ d_4^2 \succ$ \fbox{$x_4^1$},\\
      & & $X_2^2 \colon$ & $d_3^2\succ x_2^2 \succ$ \fbox{$x_2^1$}, & $X_3^2 \colon$ & $d_4^1 \succ x_3^2 \succ$ \fbox{$x_3^1$}. &&\hfill (of \cref{ex:NP-c-l-4})~$\diamond$
    \end{tabular}
    \end{table*}

\end{example}

    \noindent The underlying graph is complete. Nevertheless, only the edges as depicted in \cref{fig:example-NP-c-l-4} turn out to be relevant for swaps.

    \begin{figure*}[t!h]
      \centering
      \begin{tikzpicture}
      \def \xdistdb {1}
      \def \xdistbd {55}
      \def \ydist {.5}

      \node[agentnode] (B21) {$B_2^1$};
      \node[agentnode, below = \ydist of B21] (B22) {$B_2^2$};
      \node[agentnode, below = \ydist of B22] (B23) {$B_2^3$};
      
      \node[agentnode, right = \xdistdb of B21] (D21) {$D_2^1$};
      \node[agentnode, right = \xdistdb of B22] (D22) {$D_2^2$};
      \node[agentnode, right = \xdistdb of B23] (D23) {$D_2^3$};

      \gettikzxy{(B21)}{\xdta}{\ydta};
      \gettikzxy{(D21)}{\xbta}{\ybta};
      \gettikzxy{(B22)}{\xdtb}{\ydtb};
      \gettikzxy{(B23)}{\xdtc}{\ydtc};

      \def \ydist {30}
      
      \node[agentnode] at (\xdta*0.2+\xbta*0.8, \ydta+\ydist) (A2) {$A_2$};

      \node[agentnode] at (\xdta, \ydtc-\ydist*1.5) (X21) {$X_2^1$};
      \node[agentnode] at (\xbta, \ydtc-\ydist*1.5) (X22) {$X_2^2$};
      
      \def \ydist {.5}

      \node[agentnode] at (\xdta-\xdistbd, \ydta/2+\ydtb/2) (D11) {$D_1^1$};
      \node[agentnode] at (\xdta-\xdistbd, \ydtc/2+\ydtb/2) (D12) {$D_1^2$};

      \def \ydist {.5}

      \node[agentnode, left = \xdistdb of D11] (B11) {$B_1^1$};
      \node[agentnode, left = \xdistdb of D12] (B12) {$B_1^2$};

      \gettikzxy{(D11)}{\xxa}{\yya};
      \gettikzxy{(B11)}{\xxb}{\yyb};

      \def \ydist {30}
      
      \node[agentnode] at (\xxa*0.8+\xxb*0.2, \yya+\ydist) (A1) {$A_1$};
      
      \node[agentnode] at (\xxa*0.5+\xxb*0.5, \ydtc-\ydist*1.5) (X11) {$X_1^1$};
      
      \gettikzxy{(D21)}{\xbta}{\ybta};
      \gettikzxy{(D22)}{\xbtb}{\ybtb};
      \gettikzxy{(D23)}{\xbtc}{\ybtc};

      \def \ydist {.5}

      \node[agentnode] at (\xbta+\xdistbd, \ybta) (B31) {$B_3^1$};
      \node[agentnode] at (\xbta+\xdistbd, \ybtb) (B32) {$B_3^2$};
      \node[agentnode] at (\xbta+\xdistbd, \ybtc) (B33) {$B_3^3$};

      \def \ydist {.5}

      \node[agentnode, right = \xdistdb of B31] (D31) {$D_3^1$};
      \node[agentnode, right = \xdistdb of B32] (D32) {$D_3^2$};
      \node[agentnode, right = \xdistdb of B33] (D33) {$D_3^3$};

      \gettikzxy{(D31)}{\xxa}{\yya};
      \gettikzxy{(B31)}{\xxb}{\yyb};

      \def \ydist {30}
      
      \node[agentnode] at (\xxa*0.8+\xxb*0.2, \yya+\ydist) (A3) {$A_3$};
      \node[agentnode] at (\xxb, \ydtc-\ydist*1.5) (X31) {$X_3^1$};
      \node[agentnode] at (\xxa, \ydtc-\ydist*1.5) (X32) {$X_3^2$};

      \def \ydist {.5}
      
      \gettikzxy{(D31)}{\xbta}{\ybta};
      \gettikzxy{(D32)}{\xbtb}{\ybtb};
      \gettikzxy{(D33)}{\xbtc}{\ybtc};
      \node[agentnode] at (\xbta+\xdistbd, \ybta/2+\ybtb/2) (B41) {$B_4^1$};
      \node[agentnode] at (\xbta+\xdistbd, \ybtb/2+\ybtc/2) (B42) {$B_4^2$};

      \def \ydist {.5}

      \node[agentnode, right = \xdistdb of B41] (D41) {$D_4^1$};
      \node[agentnode, right = \xdistdb of B42] (D42) {$D_4^2$};

      \gettikzxy{(D41)}{\xxa}{\yya};
      \gettikzxy{(B41)}{\xxb}{\yyb};

      \def \ydist {30}
      
      \node[agentnode] at (\xxa*0.8+\xxb*0.2, \yya+\ydist) (A4) {$A_4$};

      \node[agentnode] at (\xxa*0.5+\xxb*0.5, \ydtc-\ydist*1.5) (X41) {$X_4^1$};
      
       \def \ydist {.5}
      
      \gettikzxy{(D41)}{\xbta}{\ybta};
      \gettikzxy{(D42)}{\xbtb}{\ybtb};

      \node[agentnode] at (\xbta+\xdistbd, \ybta/2+\ybtb/2) (I) {${\,\,}I^{\;\;}$};


      \foreach \j in {1,4} {
        \foreach \i in {1,2} {
          \draw (B\j\i) -- (D\j\i);
          \draw (A\j) -- (B\j\i);
        }
      }
      \foreach \j in {2,3} {
        \foreach \i in {1,2,3} {
          \draw (B\j\i) -- (D\j\i);
          \draw (A\j) -- (B\j\i);
        }
      }

      \foreach \i in {11,12} {
        \foreach \j in {21,22,23} {
          \draw (D\i) -- (B\j);
        }
      }
      
      \foreach \i in {21,22,23} {
        \foreach \j in {31,32,33} {
          \draw (D\i) -- (B\j);
        }
      }
      
      \foreach \i in {31,32,33} {
        \foreach \j in {41,42} {
          \draw (D\i) -- (B\j);
        }
      }
      \draw (I) -- (D41);
      \draw (I) -- (D42);

      \foreach \i / \j in {11/21, 11/31, 21/11,21/22, 22/32, 31/12,31/23,32/41,41/33,41/42} {
        \draw[red] (X\i) -- (D\j);
      }

      \draw[red] (X21) -- (X22);
      \draw[red] (X31) -- (X32);
      
      \def \soff {.3}
      \draw[gray, rounded corners]
      ($(A1.north west)+(0,\soff)$) --
      ($(B11.north west)+(-\soff,0)$) --
      ($(B12.south west)+(-\soff, -\soff)$) --
      ($(D12.south east)+(\soff, -\soff)$) --
      ($(D11.north east)+(\soff, 0)$) --
      ($(A1.north east)+(\soff*.3, \soff)$) --
      cycle;

      \draw[gray, rounded corners]
      ($(A2.north west)+(0,\soff)$) --
      ($(B21.north west)+(-\soff,0)$) --
      ($(B23.south west)+(-\soff, -\soff)$) --
      ($(D23.south east)+(\soff, -\soff)$) --
      ($(D21.north east)+(\soff, 0)$) --
      ($(A2.north east)+(\soff*.3, \soff)$) --
       cycle;

      \draw[gray, rounded corners]
      ($(A3.north west)+(0,\soff)$) --
      ($(B31.north west)+(-\soff,0)$) --
      ($(B33.south west)+(-\soff, -\soff)$) --
      ($(D33.south east)+(\soff, -\soff)$) --
      ($(D31.north east)+(\soff, 0)$) --
      ($(A3.north east)+(\soff*.3, \soff)$) --
      cycle;

      \draw[gray, rounded corners]
      ($(A4.north west)+(0,\soff)$) --
      ($(B41.north west)+(-\soff,0)$) --
      ($(B42.south west)+(-\soff, -\soff)$) --
      ($(D42.south east)+(\soff, -\soff)$) --
      ($(D41.north east)+(\soff, 0)$) --
      ($(A4.north east)+(\soff*.3, \soff)$) --
      cycle;
    \end{tikzpicture}
    \caption{Underlying graph with only relevant edges for the profile constructed according to the reduction given in \cref{thm:NP-c-length-4}}
    \label{fig:example-NP-c-l-4}
  \end{figure*}
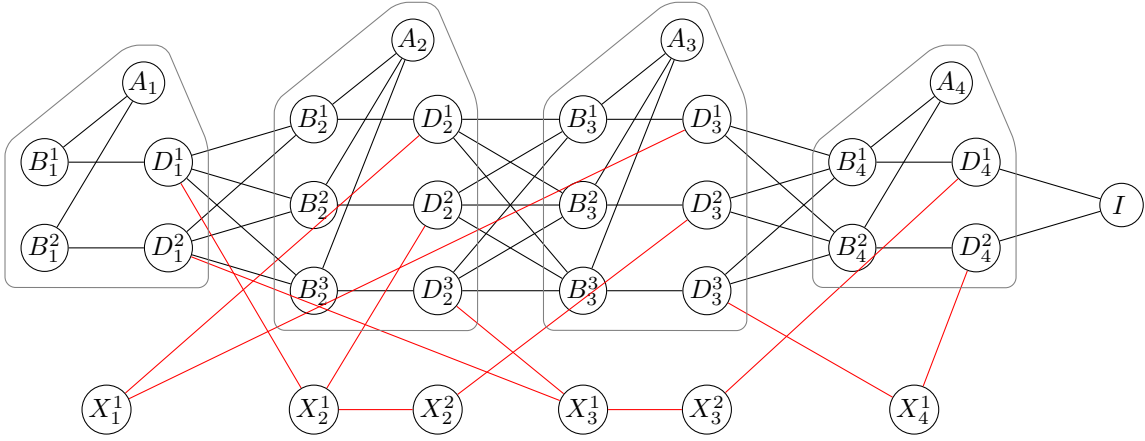
  
    Now, we show that instance~$\phi=(\mathcal{V}, \mathcal{C})$, with $n$ variables~$\mathcal{V}$ and $m$ clauses~$\mathcal{C}$, admits a satisfying truth assignment if and only if object~$x$ which agent~$A_1$ initially holds is reachable for our agent~$I$ which initially holds $a_m$.

    For the ``only if'' part, assume that $\beta\colon \mathcal{V}\to \{\ttrue,\tfalse\}$ is a satisfying assignment for~$\phi$.
    Intuitively, this satisfying assignment will guide us to find a sequence of swaps, making object~$x$ reach agent~$I$.
    
    First, for each variable~$v_i\in \mathcal{V}$, if $\occ(i)=3$,
    implying that there are two variable agents ($X_i^1$ and $X_i^2$) for~$v_i$,
    and if $\beta(v_i)=\ttrue$, then let agent~$X_i^1$ and $X_i^2$ swap their initial objects so that $X_i^1$ and $X_i^2$ hold $x_i^1$ and $x_i^2$, respectively.

    For each clause~$C_j$, identify a literal, say~$\ell_j$, which satisfies $C_j$ under~$\beta$,
    and do the following.
    \begin{compactenum}
      \item Let agent~$A_j$ and agent~$B_j^{f_j(\ell_j)}$ swap their initial objects.
      \item Let agent~$D_j^{f_j(\ell_j)}$ and agent~$X_i^z$ swap their current objects such that
      \begin{compactenum}[(a)]
        \item if $\ell_j = \negv_i$, then $z=1$ (note that in this case agent~$X_i^1$
        is holding object~$x_i^1$ if $\occ(i)=2$, and is holding object~$x_i^2$ if $\occ(i)=3$),
        \item if $\ell_j = v_i$ and $j=\pi_1(i)$, then $z=1$ (note that in this case agent~$X_i^1$ is holding object~$x_i^1$), and
        \item if $\ell_j = v_i$ and $j=\pi_2(i)$, then $z=2$ (note that in this case agent~$X_i^2$ is holding object~$x_i^2$).
      \end{compactenum}
    \end{compactenum}
    After these swaps,
    agent~$B_1$ is holding object~$x$.
    Each agent~$B_j^{f_j(\ell_j)}$, $2\le j \le m$, is holding object~$a_{j-1}$.
    Each agent~$D_j^{f_j(\ell_j)}$ is holding object~$\tau(C_j, \ell_j)$.
    Now, to let $x$ reach agent~$I$, we let it pass through agents~$B_1$, $D_1$, $B_2$, $D_2$, $\ldots$, $B_m$, $D_m$, and finally to $I$.
    Formally, iterating from $j=1$ to $j=m-1$, we do the following:
    \begin{compactenum}
      \item Let agent~$B_j$ and $D_j$ swap their current objects so that $D_j$ holds object~$x$.
      \item Let agent~$D_j$ and $B_{j+1}$ swap their current objects so that $B_{j+1}$ holds object.
    \end{compactenum}
    After these swaps, agent~$B_m$ obtains object~$x$.
    Let agent~$D_m$ swap its object~$\tau(D_m, \ell_m)$ with agent~$B_m$ for object~$x$.
    Finally, let agent~$I$ swap its object~$a_m$ with agent~$D_m$ for object~$x$.
    This completes the proof for the ``only if'' part.

    For the ``if'' part, assume that there is a sequence of swaps ~$(\sigma_0,\sigma_1,\ldots, \sigma_s)$ which makes object~$x$ reach agent~$I$, i.e.\ $\sigma_s(I)=x$.
    Now, we show how to construct a satisfying truth assignment for $\phi$.
    First, we observe the following properties which will help us to identify a literal for each clause such that setting it to \texttt{true} will satisfy the clause.

    \begin{claim}\label{claim:clause-literal}
      For each clause~$C_j\in \mathcal{C}$,
      there exist an assignment~$\sigma_r$, $1\le r < s$, and a literal~$\ell_j \in C_j$ such that $\sigma_r$ admits a swap for agents~$B_j^{f_j(\ell_j)}$ and $D_j^{f_j(\ell_j)}$,
      i.e.
      \begin{compactenum}[(1)]
        \item $\sigma_r(B_j^{f_j(\ell_j)}) = x$,
        \item $\sigma_r(D_j^{f_j(\ell_j)}) = \tau(C_j, \ell_j)$,
        \item $\sigma_{r+1}(B_j^{f_j(\ell_j)}) = \tau(C_j, \ell_j)$, and
        \item $\sigma_{r+1}(D_j^{f_j(\ell_j)}) = x$.
      \end{compactenum}
    \end{claim}

    \begin{proof}
   \renewcommand{\qedsymbol}{(of \cref{claim:clause-literal})~$\diamond$}       
    In our initial assignment~$\sigma_0$, agent~$I$ holds object~$a_m$.
    To make object~$x$ reach agent~$I$, one can verify that 
    agent~$I$ must have swapped with some clause agent~$D_m^z$ with~$z\in \{1,\ldots, |C_m|\}$ since agent~$I$ only prefers $x$ to $a_m$,
    and only agents from $\{D_m^t \mid 1\le t \le |C_m|\}$ are willing to swap $x$ for $a_m$. 
    Let $\ell_m$ be the literal with $f_m(\ell_m)=z$; recall that $f_m$ is a bijection.
    In order to make agent~$D_m^z$ obtain object~$x$, by her preference list, she must be holding object~$\tau(C_m, \ell_m)$ and swap it for $x$ since no agent will swap with her for $d_m^z$.
    Observe that agent~$B_m^z$ is the only agent that prefers $\tau(C_m,\ell_m)$ to $x$.
    It follows that $B_m^z$ must have swapped with $D_m^z$ for object $\tau(C_m, \ell_m)$.
    This means that there must be an assignment~$\sigma_{r}$ in the sequence such that
      \begin{compactenum}[(1)]
        \item $\sigma_{r}(B_m^{z}) = x$,
        \item $\sigma_{r}(D_m^{z}) = \tau(C_m, \ell_m)$,
        \item $\sigma_{r+1}(B_m^{z}) = \tau(C_m, \ell_m)$, and
        \item $\sigma_{r+1}(D_m^{z}) = x$.
      \end{compactenum}

      Now, we show our statement through induction on the index~$j$ of the clause agents, $j\ge 2$.
      Assume that there is an assignment~$\sigma_{r_j}$ in the sequence and
      that $C_j$ contains a literal~$\ell_j$ such that
      \begin{compactenum}[(1)]
        \item $\sigma_{r_j}(B_j^{f_j(\ell_j)}) = x$,
        \item $\sigma_{r_j}(D_j^{f_j(\ell_j)}) = \tau(C_j, \ell_j)$,
        \item $\sigma_{r_j+1}(B_j^{f_j(\ell_j)}) = \tau(C_j, \ell_j)$, and
        \item $\sigma_{r_j+1}(D_j^{f_j(\ell_j)}) = x$.
      \end{compactenum}
      
      By the above assignment, it follows that agent~$B_j^{f_j(\ell_j)}$ must have swapped with some other agent for object~$x$.
      Since agent~$B^{f_j(\ell_j)}_j$ prefers $x$ only to objects~$a_{j-1}$ and $b_j^{f_j(\ell_j)}$
      and since no agent prefers~$b_j^{f_{j}(\ell_j)}$ to $x$,
      it follows that agent~$B_j^{f_j(\ell_j)}$ must have swapped with some other agent  for $x$ while holding object~$a_{j-1}$.
      Since only agents from $\{D_{j-1}^{t}\mid 1\le t \le |C_{j-1}|\}$ prefer $a_{j-1}$ to $x$, it follows that $B_j^{f_j(\ell_j)}$ must have swapped with some agent~$D_{j-1}^{z_{j-1}}$ with $z_{j-1}\in \{1,\ldots, |C_{j-1}|\}$ for object~$x$.
      Let $\ell_{j-1}$ be the literal with $f_{j-1}(\ell_{j-1})=z_{j-1}$.
      To perform such a swap, however, agent~$D_{j-1}^{z_{j-1}}$ must first obtain object~$x$.
      Similarly to the case with agent~$D_m$, agent~$D_{j-1}^{z_{j-1}}$ must once hold object~$\tau(C_{j-1}, \ell_{j-1})$ and swap it for object~$x$ since no agent will swap with her for~$d_{j-1}^{z_{j-1}}$.
      Observe that agent~$B^{z_{j-1}}_{j-1}$ is the only agent that prefers~$\tau(C_{j-1}, \ell_{j-1})$ to~$x$.
      It follows that agent~$D_{j-1}^{z_{j-1}}$, while holding object~$\tau(C_{j-1}, \ell_{j-1})$, swapped with agent~$B_{j-1}^{z_{j-1}}$ for object~$x$, i.e.\
      there is an assignment~$\sigma_{r_{j-1}}$ in the sequence and $C_{j-1}$ contains a literal~$\ell_{j-1}$ such that
      \begin{compactenum}[(1)]
        \item $\sigma_{r_{j-1}}(B_{j-1}^{f_{j-1}(\ell_{j-1})}) = x$,
        \item $\sigma_{r_{j-1}}(D_{j-1}^{f_{j-1}(\ell_{j-1})}) = \tau(C_{j-1}, \ell_{j-1})$,
        \item $\sigma_{r_{j-1}+1}(B_{j-1}^{f_{j-1}(\ell_{j-1})}) = \tau(C_{j-1}, \ell_{j-1})$, and
        \item $\sigma_{r_{j-1}+1}(D_{j-1}^{f_{j-1}(\ell_{j-1})}) = x$. \qedhere
      \end{compactenum}
  \end{proof}

  \noindent By the above claim, we can now define a truth assignment~$\beta$ for $\phi$.
  \begin{align*}
    & \text{For all } v_i \in \mathcal{V}, ~~\text{ let } \beta(v_i) \coloneqq   \begin{cases} \tfalse, &\text{if } D_j^{f_j(\negv_i)} \text{ swapped}  \text{ with } \\
      & B_j^{f_j(\negv_i)} \text{ for } x \text{, where } j=\nu(i),
      \\
      \ttrue, &\text{otherwise.}\end{cases}
  \end{align*}
  Recall that in~$\phi$ each variable appears exactly once as a negative literal.
  Thus, our $\beta$ is a well-defined truth assignment.
  To show that~$\beta$ is indeed a satisfying assignment, suppose, towards a contradiction, that $\beta$ does not satisfy some clause~$C_j \in \mathcal{C}$.
  By \cref{claim:clause-literal}, let $\ell_j\in C_j$ be a literal such that $D_j^{f_j(\ell_j)}$, while holding object~$\tau(C_j, \ell_j)$, swapped with $B_j^{f_j(\ell_j)}$ for $x$.
  Observe that $\ell_j\in \{v_i, \negv_i\}$ for some $v_i\in \mathcal{V}$.
  We distinguish between two cases, in each of which we will arrive at a contradiction.
  
    \myparagraph{Case 1: $\boldsymbol{\ell_j=\negv_i}$.} This implies that $\negv_i\in C_j$, $j=\nu(i)$. Thus, $D_j^{f_j(\negv_i)}$ swapped with $B_j^{f_j(\negv_i)}$ for $x$. By our definition of $\beta$ it follows that $\beta(v_i)=\tfalse$ which satisfies $C_j$--a contradiction.

    \myparagraph{Case 2: $\boldsymbol{\ell_j=v_i}$.} This implies that $v_i\in C_j$. Since $C_j$ is not satisfied by $\beta$ it follows that $\negv_i\notin C_j$ and $\beta(v_i)=\tfalse$.
    By our definition of $\beta$ it follows that $D_{j'}^{f_{j'}(\negv_i)}$, while holding object~$\tau(C_{j'}, \negv_i)$ swapped with $B_{j'}^{f_{j'}(\negv_i)}$ for $x$ where $j'=\nu(i)$ and $j'\neq j$.

    If $\occ(i)=2$, implying that there is exactly one variable agent, namely $X_i^1$ for $v_i$,
    then by our definition of $\tau$ it follows that $\tau(C_j, v_i)=\tau(C_{j'}, \negv_i)=x_i^1$.
    To be able to swap away object~$x_i^1$, agent~$D_{j'}^{f_{j'}(\negv_i)}$ needs to obtain it from agent~$X_i^1$
    since~$d_{j'}^{f_{j'}(\negv_i)}$ is the only object to which agent~$D_{j'}^{f_{j'}(\negv_i)}$ prefers $x_i^1$ and
    since~$X_i^1$ is the only agent who prefers~$d_{j'}^{f_{j'}(\negv_i)}$ to~$x_i^1$.
    This implies that agent~$D_{j}^{f_{j}(v_i)}$ did not obtain object~$x_i^1$ from agent~$X_i^1$, and hence did not hold object~$x_i^1$ during the whole swap sequence.
    However, since $x_i^1$ is the only object that agent~$B_j^{f_j(\ell_j)}$ prefers to~$x$,
    it follows that agent~$D_j^{f_(j)(v_i)}$ has not swapped with $B_j^{f_j(\ell_j)}$ for $x$---a contradiction to our assumption above (before the case study).    

    Analogously, if $\occ(i)=3$, implying that there are exactly two variable agents, namely $X_i^1$ and $X_i^2$ for $v_i$, then by the definition of $\tau$ it follows that $\tau(C_{j'}, \negv_i)=x_i^2$.
    By a similar reasoning as in the case of $\occ(i)=2$, it follows that $X_{i}^1$ did not swap with any other agent for object~$x_i^1$ as this would require her to swap her initial object~$x_i^2$ which she gave away for $d_{j'}^{f_{j'}(\negv_i)}$.
    Consequently, agent~$D_j^{f_j(v_i)}$ would \myemph{not} have swapped either with agent~$X_i^1$  for object~$x_i^1$ or with agent~$X_i^2$ for object~$x_i^1$--a contradiction to our initial assumption that $D_j^{f_j(v_i)}$ swapped away $\tau(C_j, x_i)$ which is either $x_i^1$ or $x_i^2$.
    This completes the ``if'' part.
  \end{proof}

  \section{Generalized Caterpillars}\label{sec:caterpillars}

We obtain NP-hardness for \ROs{} on generalized caterpillars where each hair has length at most two and only one vertex has degree larger than two.
This strengthens the NP-hardness of \ROs{} on trees~\cite{GouLesWil2017} (their constructed tree is a generalized caterpillar where each hair has length three and there is only one vertex of degree larger than two).

For the sake of completeness, we give a full proof including parts of the original proof by \citet{GouLesWil2017}.

\newcommand{\nphcatepillars}{
\ROs{} is NP-hard on generalized caterpillars where each hair has length at most two and only one vertex has degree larger than two.  
}
\begin{theorem}
  \label[theorem]{thm:cater}
  \nphcatepillars
\end{theorem}

\begin{proof}
We use the same notation as \citet{GouLesWil2017}.
Let~$\phi = (\mathcal{V,C})$ be an instance of \textsc{Two positive one negative at most 3-Sat} with variable set~${\mathcal{V} = \{v_1,\ldots, v_n\}}$ and clause set~$\mathcal{C}=\{C_1,\ldots, C_m\}$.
For each variable~$v_i$, let~$p_1^i,p_2^i$ and~$n^i$ denote the clause in which~$v_i$ occurs first as a positive literal, second as a positive literal and as a negative literal, respectively. 
Accordingly, we denote the respective literal of~$v_i$ by~$v_i^{p_1^i},v_i^{p_2^i}$ or~$\bar{v}_i^{n^i}$.

The instance of \ROs{} is constructed as follows.
For each variable~$v_i$, we add a ``variable gadget'' consisting of nine agents~$\bar{X}_i^{n^i}, X_i^{p_1^i}, X_i^{p_2^i}, D_i^1, D_i^2, P_i^1, P_i^2, N_i$, and~$H_i$.
For each clause~$C_i$, we add one ``clause agent''~$C_i$ and we add an additional agent~$T$ which starts with an object~$t$. We ask whether agent~$C_m$ can get object~$t$.
\begin{figure}
\centering
\begin{tikzpicture}[scale=0.8, every node/.style={scale=0.95}]
\node[circle, draw, label=above:$C_m$] (Cm) at (0,6) {};
\node[circle, draw, label=above:$C_{m-1}$] (Cm1) at (1,6) {};
\node (Cdots) at (2,6) {$\dots$};
\node[circle, draw, label=above:$C_{1}$] (C1) at (3,6) {};
\node[circle, draw, label=above:$T$] (T) at (4,6) {};

\node[circle, draw, label=left:$D_i^1$] (D1) at (-4.5,3) {};
\node[circle, draw, label=left:$X_i^{p_1^i}$] (X1) at (-3,3) {};
\node[circle, draw, label=left:$X_i^{p_2^i}$] (X2) at (-1.5,3) {};
\node[circle, draw, label=left:$\bar{X}_i^{n^i}$] (Xn) at (0,3) {};
\node[circle, draw, label=right:$H_i$] (H) at (1.5,3) {};
\node[circle, draw, label=right:$D_{i+1}^1$] (D1') at (3,3) {};
\node[] (dots) at (5,3) {$\dots$};

\node[circle, draw, label=left:$D_i^2$] (D2) at (-4.5,0) {};
\node[circle, draw, label=left:$P_i^1$] (P1) at (-3,0) {};
\node[circle, draw, label=left:$P_i^2$] (P2) at (-1.5,0) {};
\node[circle, draw, label=left:$N_i$] (N) at (0,0) {};
\node[circle, draw, label=right:$D_{i+1}^2$] (D2') at (3,0) {};

\draw (Cm) -- (Cm1);
\draw (Cm1) -- (Cdots);
\draw (Cdots) -- (C1);
\draw (C1) -- (T);
\draw (Cm) -- (D1);
\draw (Cm) -- (X1);
\draw (Cm) -- (X2);
\draw (Cm) -- (Xn);
\draw (Cm) -- (H);
\draw (Cm) -- (D1');
\draw (Cm) -- (dots);
\draw (D1) -- (D2);
\draw (X1) -- (P1);
\draw (X2) -- (P2);
\draw (Xn) -- (N);
\draw (D1') -- (D2');
\end{tikzpicture}
\caption{Illustration of the construction in the proof of \cref{thm:cater} showing one variable gadget (bottom) for representing the assignment of one variable and the ``clause path'' (top) that verifies whether some assignment satisfies the input formula~$\phi$.}
\label{fig:hair2}
\end{figure}
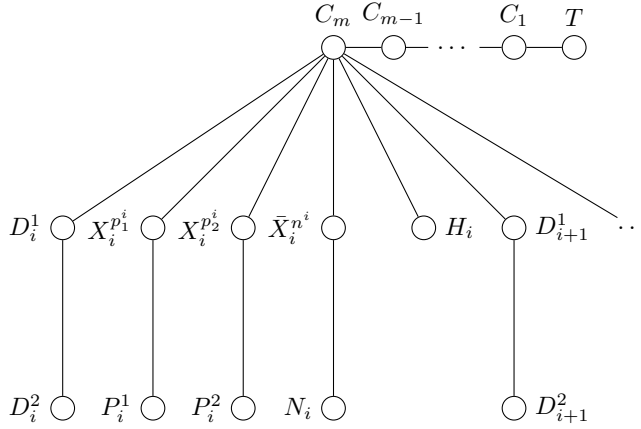
The underlying graph is depicted in \cref{fig:hair2}.
The preference lists of the agents are as follows:
  \allowdisplaybreaks[1]
\begin{align*}
D_i^1 &: v_i \succ -_i \succ \fbox{$+_i$},\\
D_i^2 &: +_i \succ \fbox{$-_i$},\\
X_i^{p_1^i} &: c_{m-p_1^i+1} \succ x_i^{p_1^i} \succ +_i \succ \fbox{$++_i$},\\
P_i^1 &: +_i \succ \fbox{$x_i^{p_1^i}$},\\
X_i^{p_2^i} &: c_{m-p_2^i+1} \succ x_i^{p_2^i} \succ ++_i \succ \fbox{$p_i$},\\
P_i^2 &: ++_i \succ \fbox{$x_i^{p_2^i}$},\\
\bar{X}_i^{n^i} &: c_{m-p_n^i+1} \succ \bar{x}_i^{n^i} \succ -_i \succ \fbox{$n_i$},\\
N_i &: -_i \succ \fbox{$\bar{x}_i^{n^i}$},\\
H_n &: p_n \succ n_n \succ \fbox{$c_m$},\\
H_i &: p_i \succ n_i \succ \fbox{$v_{i+1}$},\\
  C_m &: t \succ \{\ell_m\} \succ c_1 \succ \{\ell_{m-1}\} \succ \ldots \succ c_{m-1} \succ \{\ell_1\} \succ c_m \succ \\
      &\ \ p_n \succ n_n \succ -_n \succ ++_n \succ +_n \succ v_n \succ\\
  & \ \ p_{n-1} \succ n_{n-1} \succ \ldots \succ +_1 \succ \fbox{$v_1$},\\
C_j &: \{\ell_{j+1}\} \succ t \succ \{\ell_j\} \succ c_1 \succ \{\ell_{j-1}\} \succ \ldots \succ c_{j-1} \succ \{\ell_1\} \succ \fbox{$c_j$},\\
T &: \{\ell_1\} \succ \fbox{$t$},
\end{align*}
where~$\{\ell_j\}$ is the set of all literal objects corresponding to the literals in~$C_j$ ranked in arbitrary order. 

We start by showing that the variable gadgets work properly, that is, each variable is either set to \texttt{true} or \texttt{false} and the different literal objects can be used in an arbitrary order.
Consider a variable~$v_i$ and note the following:
\begin{compactitem}
  \item $D_i^1$ can only give either~$+_i$ or~$-_i$ to~$C_m$,
  \item $N_i$ will only release~$\bar{x}_i^{n^i}$ in exchange for~$-_i$,
  \item $P_i^1$ will only release~$x_i^{p_1^i}$ in exchange for $+_i$, and
  \item $P_i^2$ will only release~$x_i^{p_2^i}$ in exchange for~$++_i$.
\end{compactitem}
Furthermore, observe that we cannot use~$c_{m-p_1^i+1}$ to release~$++_i$ from~$X_i^{p_1^i}$ since agent~$C_m$ prefers~$c_{m-p_1^i+1}$ over~$++_i$.
Hence, agent~$C_m$ can only receive either the ``negative token''~$\bar{x}_i^{n^i}$ or positive ``token(s)''~$x_i^{p_1^i}$ and/or~$x_i^{p_2^i}$.

Once~$D_i^1$ and~$D_i^2$ have decided on whether~$+_i$ or~$-_i$ is traded in exchange for~$v_i$, agent~$C_m$ trades either with~$X_i^{p_1^i}, X_i^{p_2^i}$ and~$H_i$ or with~$\bar{X}_i^{n^i}$ and~$H_i$ in this order.
Afterwards, either~$X_i^{p_1^i}$ and~$P_i^1$ (and also~$X_i^{p_2^i}$ and~$P_i^2$) can swap or~$\bar{X}_i^{n^i}$ and~$N_i$ can swap their current objects.
Thus, after fixing all variable gadgets, agent~$C_m$ holds the object~$c_m$ and can swap it with one of the agents~$X_i^{p_1^i},X_i^{p_2^i}$ or~$\bar{X}_i^{n^i}$.

The remainder of the proof works exactly as the proof of \citeauthor{GouLesWil2017}~\cite[Theorem 1]{GouLesWil2017}.
Once each variable is assigned a truth value, the path of clause agents verifies that each clause is satisfied by this assignment.
To this end, first notice that each clause agent~$C_j$ (excluding~$C_m$) prefers~$t$ over her initial object~$c_j$ and only accepts to swap~$t$ for an object in~$\{\ell_{j+1}\}$.
Hence, the only way to move~$t$ from~$T$ to~$C_{m-1}$ involves giving each agent~$C_i$ (again excluding~$C_m$) an object associated with a literal that satisfies the clause~$C_{i+1}$.
Finally, observe that the first and last clauses also need to be satisfied in order for~$T$ and~$C_{m-1}$ to give away~$t$.
Thus, if there is a sequence of swaps such that~$C_m$ gets~$t$ in the end, then~$\phi$ is satisfiable.

If~$\phi$ is satisfiable, then there is a sequence of swaps such that~$C_m$ gets~$t$: 
We first iterate over all variable gadgets and set their value according to some satisfying assignment for~$\phi$.
Once this is done, agent~$C_m$ has object~$c_m$ and can now trade this for an object corresponding to a literal that satisfies clause~$C_1$.
This object is then passed to agent~$T$ such that~$C_m$ has object~$c_{m-1}$ which can again be swapped for an object representing a satisfying literal for clause~$C_2$.
This object is then passed to agent~$C_1$ and this procedure is repeated until~$t$ reaches~$C_m$.
\end{proof}

\section{Conclusion}

We investigated the computational complexity of \ROs{} with respect to different restrictions
regarding the underlying graph and the agent preferences.
Our work narrows the gap between known tractable and intractable cases leading to a comprehensive understanding of the complexity of \ROs{}.
In particular, we settled the complexity with respect to the preference lengths.

Several questions remain open: Can \ROs{} be
solved in polynomial time on caterpillars? Note that on stars \ROs{} can be
solved in polynomial time~\cite[Proposition~1]{GouLesWil2017}).
Also, the complexity of \ROs{} on graphs of maximum degree three
is open. \citet[Theorem 4]{SW18} showed
NP-hardness of \ROs{} on graphs of maximum degree four, while our
results imply polynomial-time solvability
of \ROs{} on graphs of maximum degree two. (Note that a graph of
maximum degree two is the disjoint union of paths and cycles,
and note that we provide an efficient algorithm for paths and sketch an adaption for cycles at the end of~\cref{sec:paths}.)
Regarding preference restrictions, following the line of studying stable matchings on restricted domains~\cite{BreCheFinNie2017}, it would be interesting to know whether assuming a special preference structure can help in finding tractable cases for our problem.
Finally, one may combine resource allocation with social welfare, measured by the egalitarian or utilitarian cost~\cite{damamme_power_2015},
and study the parameterized complexity of finding a reachable assignment which meets these criteria as studied in the context of stable matchings~\cite{CHSYicalp-par-stable2018}.

\paragraph*{Acknowledgments}
The work on this paper started at the research retreat of the Algorithmics and Computational Complexity group, TU~Berlin, held at Darlingerode, Harz, March~2018.

\bibliographystyle{named}

\end{document}